\documentclass[a4paper,10pt]{article}
\usepackage{amsmath, amsthm, amsfonts, amssymb}
\theoremstyle{plain}
\newtheorem{theorem}{Theorem}[section]
\newtheorem{lemma}{Lemma}[section]
\newtheorem{proposition}{Proposition}[section]
\newtheorem{corollary}[lemma]{Corollary}
\newtheorem*{fact}{Fact}
\theoremstyle{remark}
\newtheorem{remark}{Remark}[section]
\theoremstyle{definition}
\newtheorem*{definition}{Definition}
\numberwithin{equation}{section}
\def\Om{\Omega}
\def\e{\varepsilon}

\def\p{\partial}
\def\D{\Delta}
\def\E{\mbox{\rm e}}
\def\b{\beta}
\def\d{\delta}

\def\l{\lambda}
\def\Odr{\mathcal{O}}
\def\H{W_2}
\def\Hloc{W_{2,loc}}
\def\di{\,\mathrm{d}}
\def\iu{\mathrm{i}}
\def\la{\langle}
\def\ra{\rangle}
\DeclareMathOperator{\RE}{Re}
\DeclareMathOperator{\IM}{Im}

\DeclareMathOperator{\supp}{supp}
\DeclareMathOperator{\sgn}{sign}
\DeclareMathOperator{\conspec}{\sigma_\mathrm{e}}
\DeclareMathOperator{\resspec}{\sigma_\mathrm{r}}
\DeclareMathOperator{\pointspec}{\sigma_\mathrm{p}}
\newcommand{\Dom}{\mathfrak{D}}
\newcommand{\ie}{\emph{i.e.}}
\newcommand{\eg}{\emph{e.g.}}
\newcommand{\cf}{\emph{cf}}
\newcommand{\PT}{\mathcal{PT}}
\newcommand{\Real}{\mathbb{R}}
\newcommand{\Nat}{\mathbb{N}}
\newcommand{\Int}{\mathbb{Z}}
\newcommand{\Com}{\mathbb{C}}
\newcommand{\sii}{L_2}
\newcommand{\dist}{\mathrm{dist}}
\begin{document}
\title{\textbf{$\mathcal{PT}$-symmetric waveguides}}
\author{%
Denis Borisov$^{a,b}$ \ and \ David Krej\v{c}i\v{r}{\'\i}k$^b$%
}
\date{
\ \vspace{-4ex} \\
\small
\emph{
\begin{quote}
\begin{itemize}
\item[$a)$]
Department of Physics and Mathematics, Bashkir State Pedagogical
University, October rev. st.~3a, 450000, Ufa, Russia
\item[$b)$]
Department of Theoretical Physics,
Nuclear Physics Institute,
A\-cad\-e\-my of Sciences, 25068 \v Re\v z, Czechia
\end{itemize}
\end{quote}
}
\ \vspace{-4ex} \\
\begin{center}
\texttt{borisovdi@yandex.ru},
\texttt{krejcirik@ujf.cas.cz}
\bigskip \\
11 November 2008
\end{center}
}
\maketitle
\begin{abstract}
\noindent
We introduce a planar waveguide of constant width with
non-Hermitian $\mathcal{PT}$-symmetric Robin boundary conditions.
We study the spectrum of this system in the regime
when the boundary coupling function
is a compactly supported perturbation of a homogeneous coupling.
We prove that the essential spectrum is positive
and independent of such perturbation,
and that the residual spectrum is empty.
Assuming that the perturbation is small in the supremum norm,
we show that it gives rise to real weakly-coupled eigenvalues
converging to the threshold of the essential spectrum.
We derive sufficient conditions for these eigenvalues
to exist or to be absent.
Moreover, we construct the leading terms of the asymptotic expansions
of these eigenvalues and the associated eigenfunctions.
\medskip
\begin{itemize}
\item[\textbf{MSC\,2000:}]
35P15, 35J05, 47B44, 47B99.
\item[\textbf{Keywords:}]
non-self-adjointness, $J$-self-adjointness,
$\mathcal{PT}$-symmetry, waveguides, Robin boundary conditions,
Robin Laplacian, eigenvalue and eigenfunction asymptotics, essential
spectrum, reality of the spectrum.
\bigskip
\item[\textbf{To appear in:}]
Integral Equations and Operator Theory
\medskip \\
\verb|http://dx.doi.org/10.1007/s00020-008-1634-1|
\end{itemize}
\end{abstract}
\newpage
%---------------------%
\section{Introduction}
%---------------------%
%
There are two kinds of motivations for the present work.
The first one is due to the growing interest in spectral theory
of non-self-adjoint operators.
It is traditionally relevant to the study of dissipative processes,
resonances if one uses the mathematical tool of complex scaling,
and many others.
The most recent and conceptually new application
is based on the potential quantum-mechanical
interpretation of non-Hermitian Hamiltonians which have real spectra
and are invariant under a simultaneous
$\mathcal{P}$-parity and $\mathcal{T}$-time reversal.
For more information on the subject,
we refer to the pioneering work \cite{Bender-Boettcher_1998}
and especially to the recent review~\cite{Bender_2007} with many references.

The other motivation is due to the interesting phenomena
of the existence of bound states in quantum-waveguide systems
intensively studied for almost two decades.
Here we refer to the pioneering work~\cite{ES}
and to the reviews~\cite{DE,KKriz}.
In these models the Hamiltonian is self-adjoint
and the bound states -- often without classical interpretations --
correspond to an electron trapped inside the waveguide.

In this paper we unify these two fields of mathematical physics
by considering a quantum waveguide modelled by a non-Hermitian
$\PT$-symmetric Hamiltonian.
Our main interest is to develop a spectral theory
for the Hamiltonian and demonstrate
the existence of eigenvalues outside the essential spectrum.
For non-self-adjoint operators the location
of the various essential spectra is often
as much as one can realistically hope
for in the absence of the powerful tools
available when the operators are self-adjoint,
notably the spectral theorem and minimax principle.
In the present paper we overcome this difficulty
by using perturbation methods to study the point spectrum
in the weak-coupling regime.
In certain situations we are also able to prove
that the total spectrum is real.

Let us now briefly recall the notion of $\PT$-symmetry.
If the underlying Hilbert space of a Hamiltonian~$H$
is the usual realization of square integrable functions $\sii(\Real^n)$,
the $\PT$-symmetry invariance can be stated in terms
of the commutator relation
\begin{equation}\label{commutator}
  (\PT)H = H(\PT)  \,,
\end{equation}
where the parity and time reversal operators
are defined by $(\mathcal{P}\psi)(x):=\psi(-x)$
and $\mathcal{T}\psi:=\overline{\psi}$, respectively.
In most of the $\PT$-symmetric examples
$H$~is the Schr\"odinger operator $-\Delta+V$
with a potential~$V$ satisfying~(\ref{commutator}),
so that
$
  H^* = \mathcal{T}H\mathcal{T}
$
where~$H^*$ denotes the adjoint of~$H$.
This property is known as the $\mathcal{T}$-self-adjointness of~$H$
in the mathematical literature~\cite{Edmunds-Evans},
and it is not limited to $\PT$-symmetric Schr\"odinger operators.
More generally, given any linear operator~$H$
in an abstract Hilbert space~$\mathcal{H}$,
we understand the $\PT$-symmetry property as a special case
of the $J$-self-adjointness of~$H$:
\begin{equation}\label{Jsa}
  H^* = J H J \,,
\end{equation}
where~$J$ is a conjugation operator, \ie,
$$
  \forall \phi,\psi \in \mathcal{H} \,, \qquad
  (J\phi,J\psi)_\mathcal{H}=(\psi,\phi)_\mathcal{H} \,, \quad
  J^2\psi=\psi \,.
$$
This setting seems to be adequate
for a rigorous formulation of $\PT$-symmetric problems,
and alternative to that based on Krein spaces
\cite{Langer-Tretter_2004,Ali7}.

The nice feature of the property~(\ref{Jsa}) is that~$H$ ``is not
too far'' from the class of self-adjoint operators. In particular,
the eigenvalues are found to be real for many $\PT$-symmetric Hamiltonians
\cite{Znojil_2001,DDT,Langer-Tretter_2004,CGS,Shin,
Caliceti-Cannata-Graffi_2006,KBZ}.
However, the situation is much
less studied in the case when the resolvent of~$H$ is not compact.

The spectral analysis of non-self-adjoint operators
is more difficult than in the self-adjoint case,
partly because the residual spectrum
is in general not empty for the former.
One of the goals of the present paper is to point out
that the existence of this part of spectrum
is always ruled out by~(\ref{Jsa}):
\begin{fact}
Let~$H$ be a densely defined closed linear operator
in a Hilbert space satisfying~\eqref{Jsa}.
Then the residual spectrum of~$H$ is empty.
\end{fact}

The proof follows easily by noticing that the kernels
of $H-\lambda$ and $H^*-\overline{\lambda}$
have the same dimension \cite[Lem.~III.5.4]{Edmunds-Evans}
and by the the general fact
that the orthogonal complement of the range of a densely defined
closed operator in a Hilbert space is equal to the kernel of its adjoint.
The above result is probably not well known
in the $\PT$-symmetry community.

We continue with an informal presentation of our model
and main spectral results obtained in this paper.
The rigorous and more detailed statements are postponed
until the next section because they require a number
of technical definitions.

The Hamiltonian we consider in this paper acts as the Laplacian
in the Hilbert space of square integrable functions
over a straight planar strip and the non-Hermiticity enters
through $\PT$-symmetric boundary conditions only.
The boundary conditions are of Robin type
but with imaginary coupling.
The $\PT$-symmetric invariance then implies that
we actually deal with an electromagnetic waveguide
with radiation/dissipative boundary conditions.
In fact, the one-dimensional spectral problem
in the waveguide cross-section
has been studied recently in~\cite{KBZ} (see also~\cite{K4})
and our model can be viewed as a two-dimensional
extension of the former.

Schr\"odinger-type operators with similar non-Hermitian boundary
conditions were studied previously by Kaiser, Neidhardt and Rehberg
\cite{Kaiser-Neidhardt-Rehberg_2003a,
Kaiser-Neidhardt-Rehberg_2003b, Kaiser-Neidhardt-Rehberg_2002}.
In their papers, motivated by the needs of semiconductor physics,
the configuration space is a bounded domain and the boundary
coupling function is such that the Hamiltonian is a dissipative
operator. The latter excludes the $\PT$-symmetric models
of~\cite{KBZ} and the present paper.

The $\mathcal{T}$-self-adjointness property~(\ref{Jsa})
of our Hamiltonian is proved in Section~\ref{Sec.Domain}.
If the boundary coupling function is constant,
the spectral problem can be solved by separation of variables
and we find that the spectrum is purely essential,
given by a positive semibounded interval
(\cf~Section~\ref{Sec.un}).
In Section~\ref{Sec.ess}
we prove that the essential spectrum is stable
under compactly supported perturbations
of the coupling function.
Consequently, the essential spectrum is always real in our setting,
however, it exhibits important differences
as regards similar self-adjoint problems.
Namely, it becomes as a set independent of the value
of the coupling function at infinity
when the latter overpasses certain critical value.

In Section~\ref{Sec.point} we study the point spectrum.
We focus on the existence of eigenvalues emerging
from the threshold of the essential spectrum in the limit
when the compactly supported perturbation of the coupling function
tends to zero in the supremum norm.
It turns out that the weakly-coupled eigenvalues
may or may not exist, depending on mean values
of the local perturbation.
In the case when the point spectrum exists,
we derive asymptotic expansions of the eigenvalues
and the associated eigenfunctions.

Because of the singular nature of the $\PT$-symmetric interaction,
our example is probably the simplest non-trivial, multidimensional
$\PT$-symmetric model whatsoever
for which both the point and essential spectra exist.
We hope that the present work will stimulate more research effort
in the direction of spectral and scattering properties
of the present and other non-Hermitian $\PT$-symmetric operators.

%---------------------%
\section{Main results}\label{Sec.results}
%---------------------%
%
Given a positive number~$d$, we write $I:=(0,d)$
and consider an infinite straight strip $\Omega:=\Real \times I$.
We split the variables consistently by writing
$x=(x_1,x_2)$ with $x_1\in\Real$ and $x_2\in I$.
Let~$\alpha$ be a bounded real-valued function on~$\Real$;
occasionally we shall denote by the same symbol
the function $x\mapsto\alpha(x_1)$ on $\Omega$.
The object of our interest is the operator
in the Hilbert space $\sii(\Omega)$
which acts as the Laplacian and satisfies
the following $\PT$-symmetric boundary conditions:
\begin{equation}\label{bc}
  \partial_2\Psi + \iu\,\alpha\,\Psi
  \, = \, 0
  \qquad\mbox{on}\quad
  \partial\Omega
  \,.
\end{equation}
More precisely, we introduce
\begin{equation}\label{Hamiltonian}
  H_\alpha\Psi := -\Delta\Psi
  \,,
  \qquad
  \Psi\in\Dom(H_\alpha)
  := \left\{
  \Psi\in\H^2(\Omega) \, | \
  \Psi \ \mbox{satisfies~(\ref{bc})}
  \right\}
  \,,
\end{equation}
where the action of~$H_\alpha$ should be understood in the
distributional sense and~(\ref{bc}) should be understood in the
sense of traces~\cite{Adams}. In Section~\ref{Sec.Domain} we
show that~$H_\alpha$ is well defined in the sense that it is an
$m$-sectorial operator and that
its adjoint is easy to identify:
\begin{theorem}\label{Thm.domain}
Let $\alpha \in W_\infty^1(\mathbb{R})$. Then $H_\alpha$ is an
$m$-sectorial operator in $\sii(\Omega)$ satisfying
\begin{equation}\label{symmetry}
  H_\alpha^* = H_{-\alpha}
  \,.
\end{equation}
\end{theorem}

Of course, $H_\alpha$~is not self-adjoint
unless~$\alpha$ vanishes identically
(in this case $H_0$~is the Neumann Laplacian in $\sii(\Omega)$).
However, $H_\alpha$~is $\mathcal{T}$-self-adjoint,
\ie, it satisfies~(\ref{Jsa})
with~$J$ being the complex conjugation $\mathcal{T}:\Psi\mapsto\overline{\Psi}$.
Indeed, $H_\alpha$~satisfies the relation~(\ref{symmetry})
and it is easy to see that
\begin{equation}
  H_{-\alpha}=\mathcal{T}H_\alpha\mathcal{T}
  \,.
\end{equation}
This reflects the $\PT$-symmetry~(\ref{commutator}) of our problem,
with~$\mathcal{P}$ being defined by $(\mathcal{P}\Psi)(x):=\Psi(x_1,d-x_2)$.

An important property of an operator~$H$ in a Hilbert space~$\mathcal{H}$
being $m$-sectorial is that it is closed.
Then, in particular, the spectrum~$\sigma(H)$ is well defined
as the set of complex points~$z$ such that $H-z$ is not bijective
as the operator from $\Dom(H)$ to $\mathcal{H}$.
Furthermore, its spectrum is contained in a sector
of complex numbers~$z$ such that
$|\arg(z-\gamma)| \leqslant \theta$
with some $\gamma\in\Real$ and $\theta \in [0,\pi/2)$.
In our case, however, we are able to establish a stronger result
\begin{equation}\label{sector}
  \sigma(H_\alpha) \ \subseteq \ \Xi_\alpha :=
  \left\{
  z\in\Com \,:\,
  \RE z \geqslant 0 \,, \
  |\IM z| \leqslant 2 \, \|\alpha\|_{L_\infty(\Real)} \, \sqrt{\RE z}
  \right\}
  .
\end{equation}
This follows directly from Lemma~\ref{lm3.0a}
on which the proof of Theorem~\ref{Thm.domain} is based
(\cf~the end of Section~\ref{Sec.Domain} for more details).
Consequently, the resolvent set
$
  \rho(H_\alpha) := \Com\setminus\sigma(H_\alpha)
$
contains the complement of~$\Xi_\alpha$
and we have the bound
\begin{equation}\label{resolvent.bound}
  \|(H_\alpha - z)^{-1}\| \,\leq\, 1/\dist(z,\partial\Xi_\alpha)
  \qquad\mbox{for all}\quad
  z \in \Com\setminus\Xi_\alpha
  \,,
\end{equation}
where $\|\cdot\|$ denotes the operator norm in $\sii(\Omega)$.

Given a closed operator~$H$ in a Hilbert space~$\mathcal{H}$,
we use the following decomposition of the spectrum $\sigma(H)$:
\begin{definition}
The \emph{point} spectrum~$\pointspec(H)$
equals the set of points~$\lambda$ such that $H-\lambda$ is not injective.
The \emph{essential} spectrum~$\conspec(H)$
equals the set of points~$\lambda$ such that $H-\lambda$ is not Fredholm.
Finally, the \emph{residual} spectrum~$\resspec(H)$
equals the set of points~$\lambda$ such that $H-\lambda$ is injective
but the range of $H-\lambda$ is not dense in~$\mathcal{H}$.
\end{definition}
\begin{remark}\label{Rem.ess}
1.~The reader is warned that various other types
of essential spectra of non-self-adjoint operators
are used in the literature;
\cf~\cite[Chapt.~IX]{Edmunds-Evans} for five distinct definitions
and a detailed description of their properties.
Among them we choose that of Wolf~\cite{Wolf},
which is in general larger than that of Kato~\cite[Sec.~IV.5.6]{Kato}
based on violating the semi-Fredholm property.
(Recall that a closed operator in a Hilbert space
is called Fredholm if its range is closed
and both its kernel and its cokernel are finite-dimensional,
while it is called semi-Fredholm if its range is closed
and its kernel or its cokernel is finite-dimensional.)
However, since our operator~$H_\alpha$ is $\mathcal{T}$-self-adjoint,
the majority of the different definitions coincide
\cite[Thm~IX.1.6]{Edmunds-Evans}, in particular the two above,
and that is why we use the common notation $\conspec(\cdot)$ in this paper.
Then our choice also coincides with the definition of ``continuous spectrum''
as used for instance in the Glazman's book~\cite{Glazman}.
\smallskip \\
2.~We indeed have the decomposition (\cf~\cite[Sec.~I.1.1]{Glazman})
$$
  \sigma(H)
  = \pointspec(H) \cup \conspec(H) \cup \resspec(H)
  \,,
$$
but note that there might be intersections on the right hand side.
In particular, $\conspec(H)$ contains eigenvalues
of infinite geometric multiplicity.
\smallskip \\
3.~On the other hand, the definitions of point and residual spectra
are standard and they form disjoint subsets of $\sigma(H)$.
Recalling the general fact \cite[Sec.~V.3.1]{Kato}
that the orthogonal complement of the range of a densely defined
closed operator in a Hilbert space is equal to the kernel of its adjoint,
we obtain the following characterization of the residual spectrum
in terms of the point spectrum of the operator and its adjoint:
\begin{equation}\label{Hilbert.property}
  \resspec(H) =
  \left\{
  \lambda\in\Com\ | \
  \bar{\lambda}\in\pointspec(H^*)
  \ \& \
  \lambda\not\in\pointspec(H)
  \right\}
  .
\end{equation}
\end{remark}

The $\mathcal{T}$-self-adjointness of~$H_\alpha$
immediately implies:
\begin{corollary}\label{Corol.residual}
Suppose the hypothesis of Theorem~\ref{Thm.domain}.
Then
\begin{equation*}
  \resspec(H_\alpha)=\varnothing \,.
\end{equation*}
\end{corollary}
\begin{proof}
We repeat the proof sketched in Introduction.
Since~$H_\alpha$ is $\mathcal{T}$-self-adjoint, it is easy to see that
$\lambda$~is an eigenvalue of~$H_\alpha$ (with eigenfunction~$\Psi$)
if, and only if,
$\bar{\lambda}$~is an eigenvalue of~$H_\alpha^*$
(with eigenfunction~$\overline{\Psi}$).
It is then clear from the general identity~\eqref{Hilbert.property}
that the residual spectrum of~$H_\alpha$ must be empty.
\end{proof}

The case of uniform boundary conditions, \ie~when $\alpha$
equals identically a constant~$\alpha_0$, can be solved by
separation of variables (\cf~Section~\ref{Sec.un}).
We find
\begin{equation}\label{separation}
  \sigma(H_{\alpha_0})
  = \conspec(H_{\alpha_0})
  = [\mu_0^2,+\infty)
  \,,
\end{equation}
where the threshold~$\mu_0^2$, with the notation
\begin{equation}\label{threshold}
  \mu_0=
  \begin{cases}
    \alpha_0
    & \mbox{if} \quad |\alpha_0| \leqslant \pi/d \,,
    \\
    \pi/d
    & \mbox{if} \quad |\alpha_0| > \pi/d \,,
  \end{cases}
\end{equation}
denotes the bottom of the spectrum of the ``transverse'' operator
\begin{equation}\label{trans.op}
\begin{aligned}
  -\Delta_{\alpha_0}^I\psi &:= -\psi'' ,
  \\
  \psi \in \Dom(-\Delta_{\alpha_0}^I) &:=
  \left\{
  \psi\in\H^2(I) \, |\
  \psi'+\iu\alpha_0\psi=0 \quad \mbox{at} \quad \partial I
  \right\} .
\end{aligned}
\end{equation}

The operator~$-\Delta_{\alpha_0}^I$ was studied in~\cite{KBZ}.
Its spectrum is purely discrete and real:
\begin{equation}\label{spectrum}
  \sigma(-\Delta_{\alpha_0}^I)
  = \{\mu_j^2\}_{j=0}^\infty
  \,,
\end{equation}
where~$\mu_0$ has been introduced in~(\ref{threshold}),
\begin{equation*}
  \mu_1=
  \begin{cases}
    \alpha_0
    & \mbox{if} \quad |\alpha_0| > \pi/d \,,
    \\
    \pi/d
    & \mbox{if} \quad |\alpha_0| \leqslant \pi/d \,,
  \end{cases}
  \qquad\mbox{and}\qquad
  \mu_j:=\pi j/d
  \quad\mbox{for} \quad j\geqslant 2 \,.
\end{equation*}
Making the hypothesis
\begin{equation}\label{hypothesis}
  \alpha_0 d/\pi \ \not\in \ \Int\!\setminus\!\{0\}
  \,,
\end{equation}
the eigenvalues of~$-\Delta_{\alpha_0}^I$ are simple
and the corresponding set of eigenfunctions $\{\psi_j\}_{j=0}^\infty$
can be chosen as
\begin{equation}\label{psi}
  \psi_j(x_2) :=
  \cos(\mu_j x_2) - \iu \, \frac{\alpha_0}{\mu_j}\, \sin(\mu_j x_2)
  \,.
\end{equation}
We refer to Section~\ref{Sec.trans} for more results
about the operator $-\Delta_{\alpha_0}^I$.

Let us now turn to the non-trivial case
of variable coupling function~$\alpha$.
Among a variety of possible situations,
in this paper we restrict the considerations
to local perturbations of the uniform case.
Namely, we always assume that the difference $\alpha-\alpha_0$
is compactly supported.

First of all, in Section~\ref{Sec.ess} we show
that the essential component of the spectrum of~$H_\alpha$
is stable under the local perturbation of the uniform case:
\begin{theorem}\label{Thm.ess}
Let $\alpha-\alpha_0\in C_0(\mathbb{R})\cap W_\infty^1(\Real)$
with $\alpha_0\in\Real$. Then
\begin{equation*}
  \conspec(H_\alpha) = [\mu_0^2,+\infty)
  \,.
\end{equation*}
\end{theorem}

Notice that the essential spectrum as a set
is independent of~$\alpha_0$ as long as $|\alpha_0| \geqslant \pi/d$.
This is a consequence of the fact that
our Hamiltonian is not Hermitian.
On the other hand, it follows that the essential spectrum is real.
Recall that the residual spectrum
is always empty due to Corollary~\ref{Corol.residual}.
We do not have the proof of the reality for the point spectrum,
except for the particular case treated in the next statement:
\begin{theorem}\label{th1.2a}
Let $\alpha\in C_0(\mathbb{R})\cap W_\infty^1(\Real)$ be an odd
function. Then
\begin{equation*}
  \pointspec(H_\alpha)\subset \mathbb{R} \,.
\end{equation*}
\end{theorem}
\noindent
Summing up, under the hypotheses of this theorem
the total spectrum is real
(and in fact non-negative due to~\eqref{sector}).

The next part of our results concerns the behavior of the point
spectrum of~$H_\alpha$ under a small perturbation of~$\alpha_0$.
Namely, we consider the local perturbation of the form
\begin{equation}\label{1.2b}
  \alpha(x_1)=\alpha_0+\e\,\b(x_1) \,,
\end{equation}
where $\b\in C_0^2(\mathbb{R})$ and~$\e$ is a small positive
parameter. In accordance with Theorem~\ref{Thm.ess}, in this case
the essential spectrum of~$H_\alpha$ coincides with
$[\mu_0^2,+\infty)$, and this is also the spectrum of
$H_{\alpha_0}$. Our main interest is focused on the existence and
asymptotic behavior of the eigenvalues emerging from the
threshold~$\mu_0^2$ due to the perturbation of~$H_{\alpha_0}$
by~$\e\b$.

First we show that the asymtotically Neumann case
is in some sense exceptional:
\begin{theorem}\label{th2.0}
Suppose $\alpha_0=0$. Let~$\alpha$ be given by~\eqref{1.2b}, where
$\b\in C_0^2(\mathbb{R})$. Then the operator $H_\alpha$ has no
eigenvalues converging to $\mu_0^2$ as $\e\to+0$.
\end{theorem}

The problem of existence of the weakly-coupled eigenvalues
is more subtle as long as $\alpha_0\not=0$.
To present our results in this case,
we introduce an auxiliary sequence of functions
$v_j:\Real\to\Real$ by
\begin{equation}\label{vjs}
  v_j(x_1) :=\left\{
  \begin{aligned}
    &-\frac{1}{2}\int_\mathbb{R} |x_1-t_1| \b(t_1) \di t_1
    & \text{if} \quad j=0 \,,
    \\
    &\frac{1}{2\sqrt{\mu_j^2-\mu_0^2}}
    \int_\mathbb{R}\E^{-\sqrt{\mu_j^2-\mu_0^2}|x_1-t_1|}
    \b(t_1) \di t_1
    & \text{if} \quad j \geqslant 1 \,.
  \end{aligned}\right.
\end{equation}
Denoting $\la f\ra=\int_\mathbb{R} f(x_1)\di x_1$
for any $f \in L_1(\Real)$,
we introduce a constant~$\tau$,
depending on~$\b$, $d$ and~$\alpha_0$, by
\begin{equation*}
  \tau :=\left\{
  \begin{aligned}
  &2\alpha_0^2\la\b v_0\ra
  +\frac{2\alpha_0}{d}
  \sum\limits_{j=1}^{\infty}
  \frac{\mu_j^2\la\b v_j\ra}{\mu_j^2-\mu_0^2}
  \, \tan\frac{\alpha_0 d+j\pi}{2} \quad
  & \text{if} \quad
  |\alpha_0| < \frac{\pi}{d}
  \,,
  \\
  &\frac{2\alpha_0\pi^2\cot\frac{\alpha_0 d}{2}}{(\mu_1^2-\mu_0^2)d^3}
  \la\b v_1\ra + \frac{8 \pi^2}{(\mu_1^2-\mu_0^2)d^4}
  \sum\limits_{j=1}^{\infty} \frac{\mu_{2j}^2\la \b v_{2j}
  \ra}{\mu_{2j}^2-\mu_1
  ^2} \quad
  & \text{if} \quad
  |\alpha_0| > \frac{\pi}{d}
  \,.
  \end{aligned}\right.
\end{equation*}
It will be shown in Section~\ref{Sec.main} that the series converge.
Finally, we denote $\Om_a:=\Om\cap\{x: |x_1|<a\}$ for any
positive~$a$. Now we are in a position to state our main results
about the point spectrum.

\begin{theorem}\label{th2.1}
Suppose $|\alpha_0|<\pi/d$. Let~$\alpha$ be given by~\eqref{1.2b},
where $\b\in C_0^2(\mathbb{R})$.
\begin{enumerate}
\item
If $\alpha_0\la\b\ra<0$, there exists the unique eigenvalue
$\lambda_\e$ of $H_\alpha$ converging to $\mu_0^2$ as $\e\to+0$.
This eigenvalue is simple and real, and satisfies the asymptotic
formula
\begin{equation*}
\lambda_\e=\mu_0^2-\e^2\alpha_0^2\la\b\ra^2 +2\e^3\alpha_0\tau\la \b
\ra +\Odr(\e^4) \,.
\end{equation*}
The associated eigenfunction~$\Psi_\e$
can be chosen so that it satisfies the asymptotics
\begin{equation}\label{1.3}
\Psi_\e(x)=\psi_0(x_2)+\Odr(\e)
\end{equation}
in $\H^2(\Om_a)$ for each $a>0$,
and behaves at infinity as
\begin{equation}\label{1.4}
\Psi_\e(x)=\E^{-\sqrt{\mu_0^2-\lambda_\e}|x_1|}\psi_0(x_2)+
\Odr(\E^{-\sqrt{\mu_0^2-\lambda_\e}|x_1|}) \,,
\quad |x_1| \to +\infty \,.
\end{equation}
\item
If $\alpha_0\la\b\ra>0$, the operator $H_\alpha$ has no eigenvalues
converging to $\mu_0^2$ as $\e\to+0$.
\item
If $\la\b\ra=0$, and $\tau>0$, there exists the unique eigenvalue
$\lambda_\e$ of $H_\alpha$ converging to $\mu_0^2$ as $\e\to+0$.
This eigenvalue is simple and real, and satisfies the asymptotics
\begin{equation}\label{1.6}
\lambda_\e=\mu_0^2-\e^4\tau^2+\Odr(\e^5) \,.
\end{equation}
The associated eigenfunction can be chosen so that the relations
\eqref{1.3} and \eqref{1.4} hold true.

\item
If $\la\b\ra=0$, and $\tau<0$, the operator $H_\alpha$ has no
eigenvalues converging to~$\mu_0^2$ as $\e\to+0$.
\end{enumerate}
\end{theorem}

\begin{theorem}\label{th2.2}
Suppose $|\alpha_0|>\pi/d$ and~\eqref{hypothesis}. Let~$\alpha$ be
given by~\eqref{1.2b} where $\b\in C_0^2(\mathbb{R})$.
\begin{enumerate}
\item
If $\tau>0$, there exists the unique eigenvalue $\lambda_\e$ of
$H_\alpha$ converging to $\mu_0$ as $\e\to+0$, it is simple and
real, and satisfies the asymptotics~\eqref{1.6}. The associated
eigenfunction can be chosen so that it obeys~\eqref{1.3}
and~\eqref{1.4}.

\item
If $\tau<0$, the operator $H_\alpha$ has no eigenvalues converging
to $\mu_0^2$ as $\e\to+0$.
\end{enumerate}
\end{theorem}

In accordance with Theorem~\ref{th2.1}, in the case
$|\alpha_0|<\pi/d$ the existence of a weakly-coupled eigenvalue is
determined by the sign of the constant~$\alpha_0$ and that of the
mean value of~$\b$. In the language of Schr\"odinger operators
(treating~$\alpha$ as a singular potential), it means that a given
non-trivial~$\beta$ plays the role of an effective interaction,
attractive or repulsive depending upon the sign of~$\alpha_0$. It is
instructive to compare this situation with a self-adjoint
waveguide~\cite{BGRS}, where a similar effective interaction is
induced by a local deformation of the boundary. If the boundary is
deformed ``outward in the mean'', a weakly-coupled bound state
exists, while it is absent if the deformation is ``inward-pointing
in the mean''.

As usual, the critical situation $\langle\beta\rangle=0$
is much harder to treat.
In our case, one has to check the sign of~$\tau$
to decide whether a weakly-coupled bound state exists.
However, it can be difficult to sum up
the series in the definition of~$\tau$. This is why in our next
statement we provide a sufficient condition guaranteeing that $\tau>0$.
\begin{proposition}\label{th2.3}
Suppose $0<|\alpha_0|<\pi/d$. Let~$\alpha$ be given by~\eqref{1.2b}
where $\b(x_1)=\widetilde{\b}\left(x_1/l\right)$, $\widetilde{\b}\in
C_0^2(\mathbb{R})$, $\la\widetilde{\b}\ra =0$, $l>0$. If
\begin{equation*}
\bigg\| \int\limits_\mathbb{R}
\sgn(\cdot-t_1)\widetilde{\b}(t_1)\di
t_1\bigg\|_{L_2(\mathbb{R})}^2
  \geqslant
\frac{4\cot\frac{\alpha_0 d}{2}}{l^2 \, \alpha_0 d}
\left[\frac{\mu_1^2}{(\mu_1^2-\mu_0^2)^2} +\frac{d^2}{16\pi^2}
+\frac{d^2}{48}\right] \|\widetilde{\b}\|_{L_2(\mathbb{R})}^2 \,,
\end{equation*}
then $\tau>0$.
\end{proposition}

The meaning of this proposition is that for each positive
$|\alpha_0|<\pi/d$ the perturbation~$\e\b$ in the critical regime
$\langle\beta\rangle=0$ produces a weakly-coupled eigenvalue near
the threshold of the essential spectrum provided that the support
of~$\b$ is wide enough. This is in perfect agreement with the
critical situation of~\cite{BGRS}; according to higher-order
asymptotics derived in~\cite{BEGK}, here the weakly-coupled bound
state exists if, and only if, the critical boundary deformation is
smeared enough.

In the case $|\alpha_0|>\pi/d$ a sufficient condition guaranteeing
$\tau>0$ is given in
\begin{proposition}\label{th2.4}
Suppose $|\alpha_0|>\pi/d$ and \eqref{hypothesis}. Let~$\alpha$ be
given by~\eqref{1.2b} where $\b\in C_0^2(\mathbb{R})$. Let~$m$ be
the maximal positive integer such that $\mu_{2m}<|\alpha_0|$. If
\begin{equation}\label{1.8a}
\alpha_0\la\b v_1\ra\cot\frac{\alpha_0 d}{2}\geqslant
\frac{4}{d}\sum\limits_{j=1}^{m}\frac{\mu_{2j}^2 \la\b
v_{2j}\ra}{\mu_1^2-\mu_{2j}^2 }
  \,,
\end{equation}
then $\tau>0$.
\end{proposition}

In Section~\ref{Sec.final} we will show that the
inequality~(\ref{1.8a}) makes sense. Namely, it will be proved that
there exists~$\b$ such that this inequality holds true, provided
that~$\alpha_0$ is close enough to~$\mu_2$ but greater than this
value.

\begin{remark}\label{rm2.3}
It is useful to make the hypothesis~\eqref{hypothesis}, since it
implies that the ``transverse'' eigenfunctions~\eqref{psi} form a basis
(\cf~\eqref{basis}) and makes it therefore possible to
obtain a relatively simple decomposition of the resolvent
of~$H_{\alpha_0}$ (\cf~Lemma~\ref{Lem.constant2}).
However, it is rather a technical hypothesis for many
of the spectral results (\eg, Theorem~\ref{Thm.ess}).
On the other hand, it seems that the hypothesis is rather crucial
for the statement of Theorem~\ref{th2.2} and Proposition~\ref{th2.4}.

If the hypothesis~\eqref{hypothesis} is omitted and
$\alpha_0=\pi \ell/d$, with $\ell\in\mathbb{Z}\setminus\{0\}$,
the threshold of the essential spectrum is $\pi^2/d^2$.
This point corresponds to a simple eigenvalue
of the ``transverse'' operator $-\D^{I}_{\alpha_0}$ only if $|\ell|>1$,
while it is a double eigenvalue if $|\ell|=1$.
Under the hypothesis of Theorems~\ref{th2.1} and~\ref{th2.2},
the threshold is always a simple eigenvalue of $-\D^{I}_{\alpha_0}$,
and the proof of the theorems actually employs
some sort of ``non-degenerate'' perturbation theory.
In view of this, we conjecture that in the degenerate case $|\ell|=1$
two simple eigenvalues (possibly forming a complex conjugate pair)
or one double (real) eigenvalue
can emerge from the threshold of the essential spectrum
for a suitable choice of $\b$, while in the case $|\ell|>1$ there can
be at most one simple emerging eigenvalue. The question on the
asymptotic behaviour of these eigenvalues constitutes
an interesting open problem.
\end{remark}

%-----------------------------------%
\section{Definition of the operator}\label{Sec.Domain}
%-----------------------------------%
%
In this section we prove Theorem~\ref{Thm.domain}.
Our method is based on the theory of sectorial
sesquilinear forms~\cite[Sec.~VI]{Kato}.

In the beginning we assume only that~$\alpha$ is bounded.
Let~$h_\alpha$ be the sesquilinear form defined in $\sii(\Omega)$
by the domain
$
  \Dom(h_\alpha) := \H^1(\Omega)
$
and the prescription
$
  h_\alpha := h_\alpha^1 + \iu h_\alpha^2
$
with
\begin{align*}
  h_\alpha^1(\Psi,\Phi)
  &:= \int_\Omega \nabla \Psi(x) \cdot \overline{\nabla \Phi(x)} \, dx
  \,,
  \\
  h_\alpha^2(\Psi,\Phi)
  &:= \int_\Real \alpha(x_1) \, \Psi(x_1,d)\,\overline{\Phi(x_1,d)}\,dx_1
  - \int_\Real \alpha(x_1) \, \Psi(x_1,0)\,\overline{\Phi(x_1,0)}\,dx_1
  \,,
\end{align*}
for any $\Psi,\Phi \in \Dom(h_\alpha)$.
Here the dot denotes the scalar product in~$\Real^2$
and the boundary terms should be understood
in the sense of traces~\cite{Adams}.
We write $h_\alpha[\Psi] := h_\alpha(\Psi,\Psi)$
for the associated quadratic form,
and similarly for $h_\alpha^1$ and $h_\alpha^2$.

Clearly, $h_\alpha$~is densely defined.
It is also clear that the real part~$h_\alpha^1$
is a densely defined, symmetric, positive,
closed sesquilinear form (it is associated to the
self-adjoint Neumann Laplacian in~$\sii(\Omega)$). Of course,
$h_\alpha$~itself is not symmetric unless $\alpha$ vanishes
identically; however, it can be shown that it is sectorial and
closed. To see it, one can use the perturbation result
\cite[Thm.~VI.1.33]{Kato} stating that the sum of a sectorial
closed form with a relatively bounded form is sectorial and
closed provided the relative bound is less than one. In our
case, the imaginary part~$h_\alpha^2$ plays the role of the small
perturbation of~$h_\alpha^1$ by virtue of the following result.
\begin{lemma}\label{lm3.0a}
Let $\alpha \in L^\infty(\Real)$.
Then~$h_\alpha^2$ is relatively bounded
with respect to~$h_\alpha^1$, with
\begin{equation*}
  \left|h_\alpha^2[\Psi]\right|\
  \ \leqslant \
  2 \, \|\alpha\|_{L^\infty(\Real)} \,
  \|\Psi\|_{\sii(\Omega)} \,
  \sqrt{h_\alpha^1[\Psi]}
  \ \leqslant \
  \delta \, h_\alpha^1[\Psi]
  + \delta^{-1} \, \|\alpha\|_{L^\infty(\Real)}^2 \,
  \|\Psi\|_{\sii(\Omega)}^2
\end{equation*}
for all $\Psi\in\H^1(\Omega)$ and any positive number~$\delta$.
\end{lemma}
\begin{proof}
By density \cite[Thm.~3.18]{Adams}, it is sufficient to prove
the inequality for restrictions to~$\Omega$
of functions~$\Psi$ in $C_0^\infty(\Real^2)$.
Then we have
\begin{align*}
  \left|h_\alpha^2[\Psi]\right|
  = \left|
  \int_\Omega \alpha(x_1) \, \frac{\partial|\Psi(x)|^2}{\partial x_2}
  \, dx
  \right|
  \leqslant
  2 \, \|\alpha\|_{L^\infty(\Real)} \,
  \|\Psi\|_{\sii(\Omega)} \,
  \|\partial_2\Psi\|_{\sii(\Omega)}
  \,,
\end{align*}
which gives the first inequality after applying
$
  \|\partial_2\Psi\|_{\sii(\Omega)}
  \leqslant
  \|\nabla\Psi\|_{\sii(\Omega)}
$.
The second inequality then follows at once by means of
the Cauchy inequality with~$\delta$.
\end{proof}

In view of the above properties,
Theorem~VI.1.33 in \cite{Kato},
and the first representation theorem \cite[Thm.~VI.2.1]{Kato},
there exists the unique $m$-sectorial
operator~$\tilde{H}_\alpha$ in $\sii(\Omega)$ such that
$
  h_\alpha(\Psi,\Phi) = (\tilde{H}_\alpha \Psi,\Phi)
$
for all $\Psi\in\Dom(\tilde{H}_\alpha)\subset\Dom(h_\alpha)$
and $\Phi\in\Dom(h_\alpha)$, where
\begin{equation*}
  \Dom(\tilde{H}_\alpha)
  = \left\{
  \Psi\in\H^1(\Omega) |\,
  \exists F\in\sii(\Omega), \
  \forall \Phi\in\H^1(\Omega), \
  h_\alpha(\Psi,\Phi)=(F,\Phi)_{\sii(\Omega)}
  \right\}
  .
\end{equation*}
By integration by parts,
it is easy to check that if $\Psi\in\Dom(H_\alpha)$,
it follows that $\Psi\in\Dom(\tilde{H}_\alpha)$ with $F=-\Delta\Psi$.
That is, $\tilde{H}_\alpha$~is an extension of~$H_\alpha$
as defined in~\eqref{Hamiltonian}.
It remains to show that actually
$
  H_\alpha = \tilde{H}_\alpha
$
in order to prove Theorem~\ref{Thm.domain}.
However, the other inclusion holds as a direct consequence
of the representation theorem and the following result.
\begin{lemma}\label{lm3.0}
Let $\alpha \in W_\infty^1(\mathbb{R})$.
For each $F\in L_2(\Om)$, a generalized solution~$\Psi$
to the problem
\begin{equation}\label{5.0}
  \left\{
  \begin{aligned}
    -\D\Psi&=F
    & \mbox{in} & \quad \Omega \,,
    \\
    \partial_2\Psi + \iu\,\alpha\,\Psi &= 0
    & \mbox{on} & \quad \partial\Omega \,,
  \end{aligned}
  \right.
\end{equation}
belongs to $\Dom(H_\alpha)$.
\end{lemma}
\begin{proof}
For any function $\Psi\in\H^1(\Om)$,
we introduce the difference quotient
\begin{equation*}
  \Psi_\d(x):=\frac{\Psi(x_1+\d,x_2)-\Psi(x)}{\d}
  \,,
\end{equation*}
where~$\delta$ is a small real number.
By standard arguments
\cite[Ch.~III, Sec.~3.4, Thm.~3]{Mihajlov},
the estimate
\begin{equation}\label{5.1}
\|\Psi_\d\|_{L_2(\Om)}\leqslant \|\Psi\|_{\H^1(\Om)}
\end{equation}
holds true for all~$\d$ small enough.
If~$\Psi$ is a generalized solution to~(\ref{5.0}),
then~$\Psi_\d$ is a generalized solution to the problem
\begin{equation*}
  \left\{
  \begin{aligned}
    -\D\Psi_\d &= F_\d
    & \mbox{in} & \quad \Omega \,,
    \\
    \partial_2\Psi_\d + \iu\,\alpha\,\Psi_\d &= g
    & \mbox{on} & \quad \partial\Omega \,,
  \end{aligned}
  \right.
\end{equation*}
where~$g$ denotes the trace of the function $x \mapsto -\iu
\alpha_\d(x_1) \Psi(x_1+\d,x_2)$ to the boundary~$\partial\Omega$.
Using the ``integration-by-parts'' formula for the difference
quotients, $
  (F_\d,\Phi)_{L_2(\Om)}
  = -(F,\Phi_{-\d})_{L_2(\Om)}
$,
the integral identity corresponding
to the weak formulation of the boundary value
problem for~$\Psi_\d$ can be written as follows
\begin{multline*}
  h_\alpha(\Psi_\d,\Phi)
  = -(F,\Phi_{-\d})_{L_2(\Om)}
  - \iu \int_\Real \alpha_\d(x_1) \, \Psi(x_1+\d,d)\,\overline{\Phi(x_1,d)}\,dx_1
  \\
  + \iu \int_\Real \alpha_\d(x_1) \, \Psi(x_1+\d,0)\,\overline{\Phi(x_1,0)}\,dx_1
  \,,
\end{multline*}
where $\Phi \in \H^1(\Om)$ is arbitrary.
Letting $\Phi=\Psi_\d$,
and using the embedding of $\H^1(\Om)$ in $L_2(\p\Om)$,
the boundedness of~$\alpha_\d$,
Lemma~\ref{lm3.0a} and~(\ref{5.1}),
the above identity yields
\begin{equation*}
\|\Psi_\d\|_{\H^1(\Om)}\leqslant C \,,
\end{equation*}
where the constant $C$ is independent of $\d$.
Employing this estimate and proceeding as in the proof
of Item~b) of Theorem~3 in \cite[Ch.~III, Sec.~3.4]{Mihajlov},
one can show easily that $\p_1\Psi\in\H^1(\Om)$.
Hence, $\p_{11}\Psi\in L_2(\Om)$ and $\p_{12}\Psi\in L_2(\Om)$.

If follows from standard elliptic regularity theorems
(see, \eg, \cite[Ch.~IV, Sec.~2.2]{Mihajlov})
that $\Psi\in W_{2,loc}^2(\Omega)$.
Hence, the first of the equations in~(\ref{5.0})
holds true a.e. in $\Om$.
Thus, $\p_{22}\Psi=-F-\p_{11}\Psi\in L_2(\Om)$,
and therefore $\Psi\in \H^2(\Om)$.

It remains to check the boundary condition for~$\Psi$.
Integrating by parts, one has
\begin{align*}
(F,\Phi)_{\sii(\Omega)}=&h_\alpha(\Psi,\Phi)=
(-\D\Psi,\Phi)_{L_2(\Om)}
\\
&+ \int_\Real \big[
\partial_2\Psi(x_1,d) +
\iu\,\alpha(x_1)\,\Psi(x_1,d)\big]\,\overline{\Phi(x_1,d)}\,dx_1
\\
&- \int_\Real \big[
\partial_2\Psi(x_1,0) +
\iu\,\alpha(x_1)\,\Psi(x_1,0)\big]\,\overline{\Phi(x_1,0)}\,dx_1
\end{align*}
for any $\Phi\in \H^1(\Om)$.
This implies the boundary conditions because
$-\Delta\Psi=F$ a.e. in~$\Omega$
and~$\Phi$ is arbitrary.
\end{proof}

Summing up the results of this section, we get
\begin{proposition}\label{Prop.equality}
Let $\alpha \in W_\infty^1(\mathbb{R})$. Then
$
  \tilde{H}_\alpha = H_\alpha
$.
\end{proposition}

Theorem~\ref{Thm.domain} follows as a corollary of this proposition.
In particular, the latter implies
that~$H_\alpha$ is $m$-sectorial. Moreover, by the first
representation theorem, we know that the
adjoint~$\tilde{H}_\alpha^*$ is simply obtained as the operator
associated with
$
  h_\alpha^* = h_{-\alpha}
$.
This together with Proposition~\ref{Prop.equality} proves~(\ref{symmetry}).

Let us finally comment on the results~\eqref{sector} and~\eqref{resolvent.bound}.
As a direct consequence of the first inequality of Lemma~\ref{lm3.0},
we get that the numerical range of $H_\alpha(=\tilde{H}_\alpha)$,
defined as the set of all complex numbers $(H_\alpha\Psi,\Psi)_{\sii(\Omega)}$
where~$\Psi$ changes over all $\Psi\in\Dom(H_\alpha)$
with $\|\Psi\|_{\sii(\Omega)}=1$, is contained in the set~$\Xi_\alpha$.
Hence, in view of general results about numerical range
(\cf~\cite[Sec.~V.3.2]{Kato}),
the exterior of the numerical range of~$H_\alpha$ is a connected set,
and one indeed has~\eqref{sector} and~\eqref{resolvent.bound}.

%----------------------------------%
\section{The unperturbed waveguide}\label{Sec.un}
%----------------------------------%
%
In this section we consider the case of uniform boundary
conditions in the sense that~$\alpha$ is supposed to be
identically equal to a constant $\alpha_0\in\Real$. We prove the
spectral result~(\ref{separation}) by using the fact
that~$H_{\alpha_0}$ can be decomposed into a sum of the
``longitudinal'' operator~$-\Delta^\Real$, \ie~the self-adjoint
Laplacian in~$\sii(\Real)$, and the ``transversal'' operator
$-\Delta_{\alpha_0}^I$ defined in~(\ref{trans.op}).

%------------------------------------%
\subsection{The transversal operator}\label{Sec.trans}
%------------------------------------%
%
We summarize here some of the results established in~\cite{KBZ}
and refer to that reference for more details.

The adjoint of $-\Delta_{\alpha_0}^I$ is simply obtained
by the replacement $\alpha_0 \mapsto -\alpha_0$,
\ie,
$
  (-\Delta_{\alpha_0}^I)^* = -\Delta_{-\alpha_0}^I
$.
Consequently, $(-\Delta_{\alpha_0}^I)^*$ has the same
spectrum~(\ref{spectrum}) and the corresponding set of
eigenfunctions $\{\phi_j\}_{j=0}^\infty$ can be chosen as
$$
  \phi_j(x_2) := \overline{A_j \, \psi_j(x_2)}
  \,,
$$
where $\{\psi_j\}_{j=0}^\infty$ have been introduced in~(\ref{psi})
and~$A_j$ are normalization constants.
Choosing
\begin{equation*}
  A_{j_0}:=\frac{2\iu\alpha_0}{1-\exp{(-2\iu\alpha_0 d)}}
  \,, \quad
  A_{j_1}:=\frac{2 \mu_1^2}{(\mu_1^2-\alpha_0^2)d}
  \,, \quad
  A_j:=\frac{2\mu_j^2}{(\mu_j^2-\alpha_0^2)d}
  \,,
\end{equation*}
where $j \geqslant 2$, $(j_0,j_1) = (0,1)$ if $|\alpha_0|<\pi/d$ and
$(j_0,j_1) = (1,0)$ if $|\alpha_0|>\pi/d$ (if $\alpha_0=0$, the
fraction in the definition of $A_{j_0}$ should be understood as the
expression obtained after taking the limit $\alpha_0 \to 0$), we
have the biorthonormality relations
\begin{equation*}%\%label{orthonormality}
  \forall j,k\in\Nat, \quad
  (\psi_j,\phi_k)_{\sii(I)} = \delta_{jk}
\end{equation*}
together with the biorthonormal-basis-type expansion
(\cf~\cite[Prop.~4]{KBZ})
\begin{equation}\label{basis}
  \forall \psi\in\sii(I)
  \,, \qquad
  \psi
  = \sum_{j=0}^\infty (\psi,\phi_j)_{\sii(I)} \, \psi_j
  \,.
\end{equation}

Let us show that~(\ref{basis}) can be extended
to~$\sii(\Omega)$.
\begin{lemma}\label{lm3.1}
For any $\Psi\in L_2(\Om)$, the identity
\begin{equation*}%\l%abel{3.2b}
  \Psi(x) =
  \sum\limits_{j=0}^{\infty}
  \Psi_j(x_1)\,\psi_j(x_2)
  %\, %,
  \qquad\text{with}\qquad
  \Psi_j(x_1):=\big(\Psi(x_1,\cdot),\phi_j\big)_{L_2(I)}
  %\, %,
\end{equation*}
holds true in the sense of $L_2(\Om)$-norm.
\end{lemma}
\begin{proof}
In view of~(\ref{basis}), the series converges to~$\Psi$ in
$\sii(I)$ for almost every $x_1\in\Real$. We use the dominated
convergence theorem to prove that the convergence actually holds
in the norm of $\sii(\Omega)$. To do so, it is sufficient to
check that the $\sii(I)$-norm of the partial sums can be
uniformly estimated by a function from $\sii(\Real)$.

Let us introduce $\chi_j^D(x_2):=\sin(\pi j x_2/d)$ and
$\chi_j^N(x_2):=\cos(\pi j x_2/d)$ for $j \geqslant 1$, and
$\chi_0^N(x_2):=1/\sqrt{2}$. We recall that
$\{\sqrt{2/d}\,\chi_j^D\}_{j=1}^\infty$ and
$\{\sqrt{2/d}\,\chi_j^N\}_{j=0}^\infty$ form complete
orthonormal families in $\sii(I)$.
Expressing~$\psi_j$ in terms of $\chi_j^N$ and $\chi_j^D$ for $j
\geqslant 2$, and using the orthonormality,
we have ($n \geqslant 2$)
\begin{align}
  \Big\|
  \sum_{j=2}^{n} \Psi_j(x_1)\,\psi_j
  \Big\|_{L_2(I)}^2
  &\leqslant
  d\Big\|
  \sum\limits_{j=2}^{n} \Psi_j(x_1)\,\chi_j^N
  \Big\|_{L_2(I)}^2
  + d \, \alpha_0^2 \Big\|
  \sum\limits_{j=2}^{n} \Psi_j(x_1)\, \chi_j^D/\mu_j
  \Big\|_{L_2(I)}^2
  \nonumber \\
  &=
  d
  \sum_{j=2}^n|\Psi_j(x_1)|^2
  + d \, \alpha_0^2 \,
  \sum_{j=2}^n|\Psi_j(x_1)|^2/\mu_j^2
  \nonumber \\
  &\leqslant d
  \left(
  1+\frac{\alpha_0^2}{\mu_2^2}
  \right)
  \sum_{j=2}^\infty|\Psi_j(x_1)|^2\label{3.2a}
  \,.
\end{align}
Next, writing ($j \geqslant 2$)
\begin{equation*}
  \Psi_j =
  \sqrt{\frac{d}{2}} \, A_j
  \left(
  \Psi_j^{N}- \iu \frac{\alpha_0}{\mu_j} \, \Psi_j^{D}
  \right)
  \qquad\mbox{with}\qquad
  \Psi_j^\sharp(x_1):=\big(\Psi(x_1,\cdot),\chi_j^\sharp\big)_{L_2(I)}
  \,,
\end{equation*}
noticing that $|A_j| \leqslant c $ (valid for all $j \geqslant 0$)
where~$c$ is a constant depending uniquely on~$|\alpha_0|$ and~$d$,
and using the Parseval identities for $\chi_j^N$ and $\chi_j^D$,
we obtain
\begin{align}\label{3.3a}
  \sum\limits_{j=2}^{\infty} |\Psi_j(x_1)|^2
  &\leqslant
  c^2 d \sum\limits_{j=2}^{\infty}
  \left(
  |\Psi_j^{N}|^2+\frac{\alpha_0^2}{\mu_j^2} \, |\Psi_j^{D}|^2
  \right)
  \nonumber \\
  &\leqslant
  c^2 d
  \left(
  1+\frac{\alpha_0^2}{\mu_2^2}
  \right)
  \|\Psi(x_1,\cdot)\|_{L_2(I)}^2
  \,.
\end{align}
At the same time, using just the estimates $|\psi_j|^2\leqslant
(1+\alpha_0^2/\mu_j^2)$ valid for all $j \geqslant 0$, we readily
get
\begin{equation}\label{3.4b}
 \Big\|
  \sum_{j=0}^{1} \Psi_j(x_1)\,\psi_j
  \Big\|_{L_2(I)}
  \leqslant
  2\,c\,d \left(
  1+\frac{\alpha_0^2}{\mu_0^2}
  \right)
  \|\Psi(x_1,\cdot)\|_{L_2(I)}
  \,.
\end{equation}
Summing up,
\begin{equation*}
  \Big\|
  \sum\limits_{j=0}^{n} \Psi_j(x_1)\,\psi_j
  \Big\|_{L_2(I)}
  \leqslant C \, \|\Psi(x_1,\cdot)\|_{L_2(I)}
  \in L_2(\mathbb{R})
  \,,
\end{equation*}
where~$C$ is a constant independent of~$n$, and the usage of the
dominated convergence theorem is justified.
\end{proof}
%

%-------------------------------------------------%
\subsection{Spectrum of the unperturbed waveguide}
%-------------------------------------------------%
%
First we show that
the spectrum of~$H_{\alpha_0}$ is purely essential.
Since the residual spectrum is always empty
due to Corollary~\ref{Corol.residual},
it is enough to show that there are no eigenvalues.
\begin{proposition}\label{Prop.spec0.p}
Let $\alpha_0\in\Real$ satisfy~\eqref{hypothesis}. Then
$$
  \pointspec(H_{\alpha_0})=\varnothing
  \,.
$$
\end{proposition}
\begin{proof}
Suppose that~$H_{\alpha_0}$ possesses
an eigenvalue~$\lambda$ with eigenfunction~$\Psi$.
Multiplying the eigenvalue equation with~$\overline{\phi_j}$
and integrating over~$I$, we arrive at the equations
$$
  -\Psi_j''=(\lambda-\mu_j^2)\Psi_j
  \qquad\mbox{in}\qquad
  \Real
  \,,
  \qquad
  j \geqslant 0
  \,,
$$
where~$\Psi_j$ are the coefficients of Lemma~\ref{lm3.1}. Since
$\Psi_j\in\sii(\Real)$ due to Fubini's theorem, each of the
equations has just a trivial solution.
This together with Lemma~\ref{lm3.1} yields $\Psi=0$,
a contradiction.
That is, the point spectrum of~$H_{\alpha_0}$ is empty.
\end{proof}
\begin{remark}\label{Rem.holomorphic}
Regardless of the technical hypothesis~(\ref{hypothesis}),
the set of isolated eigenvalues of~$H_{\alpha_0}$ is always empty.
This follows from Proposition~\ref{Prop.spec0.p}
and the fact that $\alpha_0 \mapsto H_{\alpha_0}$ forms
a holomorphic family of operators \cite[Sec.~VII.4]{Kato}.
\end{remark}

It is well known that the spectrum of the ``longitudinal''
operator $-\Delta^\Real$ is also purely essential
and equal to the semi-axis $[0,+\infty)$.
In view of the separation of variables,
it is reasonable to expect that the (essential) spectrum
of~$H_{\alpha_0}$ will be given by that semi-axis
shifted by the first eigenvalue of $-\Delta_{\alpha_0}^I$.
First we show that the resulting interval indeed belongs
to the spectrum of~$H_{\alpha_0}$.
\begin{lemma}\label{Lem.constant1}
Let $\alpha_0\in\Real$. Then
$
  [\mu_0^2,+\infty) \subseteq \conspec(H_{\alpha_0})
$.
\end{lemma}
\begin{proof}
Since the spectrum of~$H_{\alpha_0}$ is purely essential, it can be
characterized by means of singular sequences
\cite[Thm.~IX.1.3]{Edmunds-Evans} (it is in fact an equivalent
definition of another type of essential spectrum, which is in
general intermediate between the essential spectra due to Wolf and
Kato, and therefore coinciding with them in our case,
\cf~Remark~\ref{Rem.ess}). That is, $\l\in\conspec(H_{\alpha_0})$
if, and only if, there exists a sequence
$\{u_n\}_{n=1}^\infty\subset\Dom(H_{\alpha_0})$ such that
$\|u_n\|_{\sii(\Omega)}=1$, $u_n \rightharpoonup 0$ and
$\|H_{\alpha_0}u_n-\l u_n\|_{\sii(\Omega)} \to 0$ as $n\to\infty$.
Let $\{\varphi_n\}_{n=1}^\infty$ be a singular sequence of
$-\Delta^\Real$ corresponding to a given $z\in[0,+\infty)$. Then it
is easy to verify that~$u_n$ defined by $
  u_n(x) :=
  \varphi_n(x_1)\psi_0(x_2)/\|\psi_0\|_{\sii(I)}
$
forms a singular sequence of~$H_{\alpha_0}$ corresponding
to~$z+\mu_0^2$.
\end{proof}

To get the opposite inclusion, we employ the fact
that the biorthonormal-basis-type relations~(\ref{basis})
are available. This enables us to decompose
the resolvent of~$H_{\alpha_0}$ into the transverse biorthonormal-basis.
\begin{lemma}\label{Lem.constant2}
Let $\alpha_0\in\Real$ satisfy~\eqref{hypothesis}. Then
$
  \Com\setminus[\mu_0^2,+\infty) \subseteq \rho(H_{\alpha_0})
$
and for any $z\in\Com\setminus[\mu_0^2,+\infty)$ we have
\begin{equation*}
  (H_{\alpha_0}-z)^{-1}
  = \sum_{j=0}^\infty
  (-\Delta^\Real+\mu_j^2-z)^{-1} B_{j}
  \,.
\end{equation*}
Here~$B_j$ is a bounded operator on $\sii(\Omega)$ defined by
$$
  (B_j\Psi)(x)
  := \big(\Psi(x_1,\cdot),\phi_j\big)_{\sii(I)} \, \psi_j(x_2)
  \,, \qquad
  \Psi\in\sii(\Omega)
  \,,
$$
and $(-\Delta^\Real+\mu_j^2-z)^{-1}$ abbreviates
$(-\Delta^\Real+\mu_j^2-z)^{-1} \otimes 1$ on
$
  \sii(\Real)\otimes\sii(I)
  \simeq \sii(\Omega)
$.
\end{lemma}
\begin{proof}
Put $z\in\Com\setminus[\mu_0^2,+\infty)$.
For every $\Psi\in L_2(\Om)$ and all $j \geqslant 0$ we denote
$
  U_j := (-\D^{\mathbb{R}}+\mu_j^2-z)^{-1}\Psi_j
  \in \sii(\Real)
$,
where $\Psi_j$ are defined in Lemma~\ref{lm3.1}.
It is clear that
\begin{equation}\label{3.4a}
  \|U_j\|_{L_2(\mathbb{R})}
  \leqslant
  \frac{\|\Psi_j\|_{L_2(\mathbb{R})}}{\dist\big(z,[\mu_j^2,+\infty)\big)}
  \leqslant C\,\frac{\|\Psi_j\|_{L_2(\mathbb{R})}}{j^2+1} \,,
  \quad
  \|U_j'\|_{L_2(\mathbb{R})}
  \leqslant
  C\,\frac{\|\Psi_j\|_{L_2(\mathbb{R})}}{\sqrt{j^2+1}} \,,
\end{equation}
where~$C$ is a constant depending uniquely on~$|\alpha_0|$, $d$ and~$z$;
the second inequality follows from the identity
$
  \|U_j'\|_{L_2(\mathbb{R})}^2+(\mu_j^2-z)
  \|U_j\|_{L_2(\mathbb{R})}^2=(\Psi_j,U_j)_{L_2(\mathbb{R})}
$.
Using~(\ref{3.4a}) and estimates of the type~(\ref{3.4b}),
it is readily seen that each function
$R_j:x \mapsto U_j(x_1) \psi_j(x_2)$ belongs to $\H^1(\Omega)$.
We will show that it is the case for their infinite sum too.
Firstly, a consecutive use of~(\ref{3.2a}),
the first inequality of~(\ref{3.4a}) and~(\ref{3.3a})
together with Fubini's theorem implies
$
  \|\sum_{j=2}^{n} R_j\|_{L_2(\Omega)}
  \leqslant
  K \, \|\Psi\|_{L_2(\Omega)}
$,
where~$K$ is a constant independent of~$n \geqslant 2$.
Secondly, a similar estimate for the partial sum of $\partial_1 R_j$
can be obtained in the same way,
provided that one uses the second inequality of~(\ref{3.4a}) now.
Finally, since the derivative of~$\psi_j$ as well can be expressed
in terms of~$\chi_j^N$ and~$\chi_j^D$ introduced in the proof
of Lemma~\ref{lm3.1}, a consecutive use of the estimates
of the type~(\ref{3.2a}) and the first inequality of~(\ref{3.4a})
together with Fubini's theorem implies
\begin{equation*}
  \Big\|\sum\limits_{j=2}^{n} \partial_2 R_j\Big\|_{L_2(\Omega)}^2
  \leqslant
  d \sum\limits_{j=2}^{n} (\alpha_0^2+\mu_j^2) \|U_j\|_{L_2(\Real)}^2
  \leqslant
  d \, C^2 \sum\limits_{j=2}^{n} \frac{\alpha_0^2+\mu_j^2}{(j^2+1)^2}
  \|\Psi_j\|_{L_2(\Real)}^2
  \,;
\end{equation*}
here the fraction in the upper bound forms a bounded sequence in~$j$,
so that we may continue to estimate as above using~(\ref{3.3a})
together with Fubini's theorem again.
Summing up, the series $\sum_{j=0}^\infty R_j$
converges in $\H^1(\Om)$ to a function~$R$ and
\begin{equation*}
  \|R\|_{\H^1(\Om)}\leqslant \tilde{K} \, \|\Psi\|_{L_2(\Om)}
  \,,
\end{equation*}
where~$\tilde{K}$ depends uniquely on~$|\alpha_0|$, $d$ and~$z$.
Employing this fact and the definition of $U_j$,
one can check easily that $R$ satisfies the identity
%\begin{equation*}
$
  h_{\alpha_0}(R,\Phi)-z(R,\Phi)_{L_2(\Om)}=(\Psi,\Phi)_{L_2(\Om)}
$
%\end{equation*}
for all $\Phi\in \H^1(\Om)$. It implies that
$R\in\mathfrak{D}(H_{\alpha_0})$ and $(H_{\alpha_0}-z)R=\Psi$, i.e.,
$R=(H_{\alpha_0}-z)^{-1}\Psi$.
\end{proof}

Lemmata~\ref{Lem.constant1} and~\ref{Lem.constant2} yield
\begin{proposition}\label{Prop.spec0}
Let $\alpha_0\in\Real$. Then
$$
  \conspec(H_{\alpha_0})=[\mu_0^2,+\infty)
  \,.
$$
\end{proposition}
\begin{proof}
In view of the lemmata,
the result holds for every $\alpha_0\in\Real$
except for a discrete set of points
complementary to the hypothesis~(\ref{hypothesis}).
However, these points can be included by noticing
that $\alpha_0 \mapsto H_{\alpha_0}$ forms
a holomorphic family of operators
(\cf~Remark~\ref{Rem.holomorphic}).
\end{proof}

Proposition~\ref{Prop.spec0}
and Proposition~\ref{Prop.spec0.p}
together with Remark~\ref{Rem.holomorphic}
imply that the spectrum of~$H_{\alpha_0}$ is real
and~(\ref{separation}) holds true for every $\alpha_0\in\Real$.

\begin{remark}
It follows from~(\ref{separation}) that
the spectrum of~$H_{\alpha_0}$ is equal to the sum
of the spectra of~$-\Delta^\Real$ and $-\Delta_{\alpha_0}^I$.
This result could alternatively be obtained by using
a general theorem about the spectrum
of tensor products~\cite[Thm.~XIII.35]{RS4}
and the fact that the one-dimensional operators
generate bounded holomorphic semigroups.
However, we do not use this way of proof
since Lemma~\ref{Lem.constant2} is employed not only
in the proof of~(\ref{separation})
but also in the proofs of Theorems~\ref{th2.1} and~\ref{th2.2}.
\end{remark}
%

%---------------------------------------------%
\section{Stability of the essential spectrum}\label{Sec.ess}
%---------------------------------------------%
%
In this section we show that the essential spectrum
is stable under a compactly supported perturbation
of the boundary conditions.
In fact, we will establish a stronger result,
namely that the difference of the resolvents of~$H_\alpha$
and~$H_{\alpha_0}$ is a compact operator.
As an auxiliary result, we shall need the following lemma.
\begin{lemma}\label{Lem.bvp}
Let $\alpha_0\in\Real$ and $\varphi\in\sii(\partial\Omega)$.
There exist positive constants~$c$ and~$C$,
depending on~$d$ and~$|\alpha_0|$,
such that any solution $\Psi\in\H^1(\Omega)$
of the boundary value problem
\begin{equation}\label{bvp}
\left\{
\begin{aligned}
  (-\Delta-z)\Psi &= 0
  &\mbox{in}& \quad \Omega
  \,, \\
  (\partial_2+\iu\alpha_0)\Psi &= \varphi
  &\mbox{on}& \quad \partial\Omega
  \,,
\end{aligned}
\right.
\end{equation}
with any $z \leqslant -c$,
satisfies the estimate
\begin{equation}\label{bvp-estimate}
  \|\Psi\|_{\H^1(\Omega)}
  \ \leqslant \
  C \, \|\varphi\|_{\sii(\partial\Omega)}
  \,.
\end{equation}
\end{lemma}
\begin{proof}
Multiplying the first equation of~(\ref{bvp})
by~$\overline{\Psi}$ and integrating over~$\Omega$, we arrive at
the identity
\begin{equation}\label{bvp-identity}
  \|\nabla\Psi\|^2 - z \, \|\Psi\|^2
  +\iu \alpha_0 \int_{\partial\Omega} \nu_2 \, |\Psi|^2
  - \int_{\partial\Omega} \nu_2 \, \varphi\overline{\Psi}
  = 0
  \,,
\end{equation}
where~$\nu_2$ denotes the second component
of the outward unit normal vector to~$\partial\Omega$.
Using the Schwarz and Cauchy inequalities
together with $|\nu_2|=1$, we have
\begin{align*}
  \left|
  \int_{\partial\Omega} \nu_2 \, |\Psi|^2
  \right|
  &=
  \left|
  \int_\Omega \partial_2|\Psi|^2
  \right|
  = 2 \left|
  \RE\,(\Psi,\partial_2\Psi)
  \right|
  \leqslant
  \delta^{-1} \, \|\Psi\|^2 + \delta \, \|\nabla\Psi\|^2
  \,,
  \\
  2 \left|
  \int_{\partial\Omega} \nu_2 \, \varphi\overline{\Psi}
  \right|
  &\leqslant
  \delta^{-1} \, \|\varphi\|_{\sii(\partial\Omega)}
  +
  \delta \,\|\Psi\|_{\sii(\partial\Omega)}
  \,,
\end{align*}
with any $\delta\in(0,1)$.
Here
$
  \|\Psi\|_{\sii(\partial\Omega)}
  \leqslant C \, \|\Psi\|_{\H^1(\Omega)}
$, where~$C$ is the constant coming from the embedding of
$\H^1(\Omega)$ in $\sii(\partial\Omega)$ (depending only on~$d$
in our case). Choosing now sufficiently small~$\delta$ and
sufficiently large negative~$z$, it is clear
that~(\ref{bvp-identity}) can be cast into the
inequality~(\ref{bvp-estimate}).
\end{proof}

Now we are in a position to prove
\begin{proposition}\label{compact}
Let $\alpha-\alpha_0 \in C_0(\Real) \cap W_\infty^1(\Real)$
with $\alpha_0\in\Real$.
Then
$$
  (H_\alpha-z)^{-1} - (H_{\alpha_0}-z)^{-1}
  \quad
  \mbox{is compact in $\sii(\Omega)$}
$$
for any $z\in\rho(H_\alpha)\cap\rho(H_{\alpha_0})$.
\end{proposition}
\begin{proof}
It is enough to prove the result for one~$z$
in the intersection of the resolvent sets
of~$H_\alpha$ and~$H_{\alpha_0}$,
and we can assume that the one is negative
(since the operators are $m$-accretive).
Given $\Phi\in\sii(\Omega)$, let
$
  \Psi :=
  (H_\alpha-z)^{-1}\Phi - (H_{\alpha_0}-z)^{-1}\Phi
$. It is easy to check that~$\Psi$ is the unique solution to~(\ref{bvp})
with $\varphi:= -\iu(\alpha-\alpha_0) T (H_\alpha-z)^{-1}\Phi$,
where~$T$ denotes the trace operator from
$\H^2(\Omega) \supset \Dom(H_\alpha)$ to $\H^1(\partial\Omega)$.
By virtue of Lemma~\ref{Lem.bvp},
it is therefore enough to show that
$
  (\alpha-\alpha_0) T (H_\alpha-z)^{-1}
$
is a compact operator
from $\sii(\Omega)$ to $\sii(\partial\Omega)$.
However, this property follows from the fact that
$\H^1(\partial\Omega)$ is compactly embedded in $\sii(\omega)$
for every bounded subset~$\omega$ of~$\partial\Omega$,
due to the Rellich-Kondrachov theorem \cite[Sec.~VI]{Adams}.
\end{proof}
\begin{corollary}
Suppose the hypothesis of Proposition~\ref{compact}. Then
$$
  \conspec(H_\alpha) = [\mu_0^2,+\infty)
  \,.
$$
\end{corollary}
\begin{proof}
Our definition of essential spectrum is indeed stable
under relatively compact perturbations
\cite[Thm.~IX.2.4]{Edmunds-Evans}.
\end{proof}
%

%\newpage
%-----------------------%
\section{Point spectrum}\label{Sec.point}
%-----------------------%
%
In this section we prove Theorems~\ref{th1.2a}--\ref{th2.2}
and Propositions~\ref{th2.3}--\ref{th2.4}.
In the proofs of Theorems~\ref{th2.0}--\ref{th2.2}
we follow the main ideas of~\cite{Ga}.
Throughout this section we assume that the identity~(\ref{1.2b})
and the assumption~(\ref{hypothesis}) hold true.

%-----------------------------------------%
\subsection{Proof of Theorem~\ref{th1.2a}}
%-----------------------------------------%
%
Any eigenvalue of infinite geometric multiplicity belongs to the
essential spectrum which is real by Theorem~\ref{Thm.ess}. Let $\l$
be an eigenvalue of $H_\alpha$ of finite geometric multiplicity,
and $\Psi$ be an associated eigenfunction. Using that~$\alpha$ is of compact
support, it is easy to check that $x \mapsto \Psi(-x_1,d-x_2)$ is an
eigenfunction associated with $\l$, too. The geometric multiplicity of $\l$
being finite, we conclude that at least one of the eigenfunction
associated with $\l$ satisfies $|\Psi(x)|=|\Psi(-x_1,d-x_2)|$.
Taking this identity into account, integrating by parts and using
the hypothesis that~$\alpha$ is odd, we obtain
\begin{equation*}
  \l \, \|\Psi\|_{L_2(\Omega)}^2
  = h_\alpha^{1}[\Psi]+\iu h_\alpha^{2}[\Psi]
  = \|\nabla\Psi\|_{L_2(\Omega)}^2
  \,,
\end{equation*}
which implies that $\lambda$ is real.
\hfill\qed

%-----------------------------%
\subsection{Auxiliary results}
%-----------------------------%
%
Let a function $F\in L_2(\Om)$ be such that $\supp
F\subseteq\overline{\Om_b}$ for fixed $b>0$.
We consider the boundary value problem
\begin{equation}\label{4.3}
\left\{
\begin{aligned}
  -\D U&=(\mu_0^2-k^2)U+F
  &\mbox{in }& \quad \Omega
  \,,
  \\
  \left(\p_2+\iu\alpha_0\right)U&=0
  &\mbox{on }& \quad \p\Om
  \,,
\end{aligned}
\right.
\end{equation}
where the parameter $k\in\mathbb{C}$ ranges in a small
neighbourhood of zero. The problem can be solved by
separation of variables justified in Lemma~\ref{Lem.constant2}
whenever $k^2\not\in(-\infty,0]$. Moreover, it is possible to extend
the solution of~(\ref{4.3}) analytically with respect to~$k$.
Namely, the following statement is valid.
\begin{lemma}\label{lm4.2}
For all  small $k\in \mathbb{C}$ there exists the unique
solution to \eqref{4.3} satisfying
\begin{equation}\label{4.5}
  U(x;k)=c_\pm(k)\E^{-k|x_1|}\psi_0(x_2)+
  \Odr\big(\E^{-\RE\sqrt{\mu_1^2-\mu_0^2+k^2}|x_1|}\big)
  \,,
\end{equation}
in the limit $x_1\to\pm\infty$, where $c_\pm(k)$ are constants.
The mapping $T_1(k)$ defined as $T_1(k)F:=U$ is a bounded linear
operator from $L_2(\Om_b)$ into $\H^2(\Om_a)$ for each $a>0$.
This operator is meromorphic with respect to~$k$
and has the simple pole at zero,
\begin{gather*}%\l%abel{4.6}
T_1(k)F=\frac{\psi_0}{2k}\int\limits_{\Om}
F(x)\overline{\phi_0}(x_2)\di x+T_2(k)F,
\end{gather*}
where  for each $a>0$
the operator $T_2(k): L_2(\Om_b)\to\H^2(\Om_a)$ is linear,
bounded and holomorphic with respect to~$k$ small enough.
The function $\widehat{U}:=T_2(0)F$ is the unique
solution to the problem
\begin{equation}\label{4.9}
\left\{
\begin{aligned}
  -\D \widehat{U}=\mu_0^2\;\!\widehat{U}+F
  \quad \mbox{in} \quad \Omega
  , \qquad
  \left(\partial_2+\iu\alpha_0\right)\widehat{U}=0
  \quad \mbox{on} \quad \partial\Omega
  ,
  \\
  \widehat{U}(x)=-\frac{\psi_0(x_2)}{2}\int\limits_{\Om}
  |x_1-t_1|F(t)\overline{\phi_0}(t_2)\di
  t+\Odr\big(\E^{-\sqrt{\mu_1^2-\mu_0^2}|x_1|}\big)
  , \quad |x_1|\to+\infty ,
\end{aligned}
\right.
\end{equation}
given by the formula
\begin{equation}\label{4.19}
  \widehat{U}(x)=\sum\limits_{j=0}^{\infty}
  \widehat{U}_j(x_1)\psi_j(x_2)
\end{equation}
with
\begin{equation*}
  \widehat{U}_j(x_1) :=
  \begin{cases}
    \displaystyle
    -\frac{1}{2}\int\limits_{\Om} |x_1-t_1|
    F(t)\overline{\phi_0}(t_2)\di t
    & \mbox{if} \quad j=0
    \,,
    \\
    \displaystyle
    \frac{1}{2\sqrt{\mu_j^2-\mu_0^2}}
    \int\limits_{\Om} \E^{-\sqrt{\mu_j^2-\mu_0^2}|x_1-t_1|}
    F(t)\overline{\phi_j}(t_2)\di t
    & \mbox{if} \quad j \geqslant 1
    \,.
  \end{cases}
\end{equation*}
\end{lemma}
The lemma is proved in the same way as Lemma~3.1 in~\cite{Bo}.

Let $M_\e$ be the operator of multiplication by the function
$x \mapsto \E^{-\iu\e\b(x_1)x_2}$.
It is straightforward to check that~$H_\alpha$
is unitarily equivalent to the operator
\begin{equation*}
  M_\e^{-1}H_\alpha M_\e=H_{\alpha_0}-\e L_\e
  \,,
\end{equation*}
where
\begin{equation*}
  L_\e=-2\iu\b'(x_1)x_2\frac{\p}{\p x_1}-2\iu\b(x_1)\frac{\p}{\p
  x_2}-\left(\e\b^2(x_1)+\iu\b''(x_1) x_2+\e \b'^2(x_1)
  x_2^2\right).
\end{equation*}
We observe that the coefficients of~$L_\e$
are compactly supported
and that their supports are bounded uniformly in~$\e$.

It follows that the eigenvalue equation for~$H_\alpha$ is equivalent
to
\begin{equation*}%\l%abel{4.7}
  H_{\alpha_0} U =\l U+\e L_\e U,
\end{equation*}
where an eigenfunction~$\Psi$ of~$H_\alpha$ satisfies $\Psi=M_\e U$.
It can be rewritten as
\begin{equation}\label{4.8}
  \left\{
  \begin{aligned}
  -\D U&=(\mu_0^2-k^2) U+\e L_\e U
  & \mbox{in}& \quad \Om \,,
  \\
  (\partial_2+\iu\alpha_0) U &=0
  & \mbox{on}& \quad \p\Om \,,
  \end{aligned}
  \right.
\end{equation}
where we have replaced $\l$ by $\mu_0^2-k^2$.

Now, let $\l$ be an eigenvalue for~$H_\alpha$ close to $\mu_0^2$. As
$x_1\to\pm\infty$, the solution~$U$ to~(\ref{4.8}) satisfies the
asymptotic formula~(\ref{4.5}), where $k=\sqrt{\mu_0^2-\l}$ and the
branch of the root is specified by the requirement $\RE k>0$. Such
restriction guarantees that the function~$U$ together with their
derivatives decays exponentially at infinity and thus belongs to
$\H^2(\Om)$. Hence, the set of all~$k$ for which the problem
(\ref{4.8}), (\ref{4.5}) has a nontrivial solution includes the set
of all values of~$k$ related to the eigenvalues of $H_\alpha$ by the
relation $\l=\mu_0^2-k^2$. Thus, it is sufficient to find all
small~$k$ for which a nontrivial solution to (\ref{4.8}),
(\ref{4.5}) exists and to check whether the solution belongs to
$\H^2(\Om)$. If it does, the corresponding number $\l=\mu_0^2-k^2$
is an eigenvalue of~$H_\alpha$.

We introduce the numbers
\begin{equation*}
  k_1(\e):=\frac{1}{2}\int\limits_{\Om}
  \overline{\phi_0}(x_2) (L_\e\psi_0)(x) \di x
  \,, \quad
  k_2(\e):=\frac{1}{2}\int\limits_{\Om}
  \overline{\phi_0}(x_2) \big(L_\e T_2(0) L_\e\psi_0\big)(x) \di x
  \,.
\end{equation*}
Basing on Lemma~\ref{lm4.2} and arguing in the same way as in
\cite[Sec. 2]{Ga} one can prove easily the following statement
(see also \cite[Sec. 4]{Bo}).
\begin{lemma}\label{lm4.3}
There exists the unique function $\e \mapsto k(\e)$
converging to zero as $\e\to+0$ for which the problem
\eqref{4.8}, \eqref{4.5} has a nontrivial solution.
It satisfies the asymptotics
\begin{equation*}%\l%abel{4.11}
k(\e)=\e k_1(\e)+\e^2 k_2(\e)+\Odr(\e^3).
\end{equation*}
The associated nontrivial solution to \eqref{4.8}, \eqref{4.5}
is unique up to a multiplicative constant and can be chosen so
that it obeys~\eqref{4.5} with
\begin{equation}\label{4.12}
c_\pm(k(\e))=1+\Odr(\e),\quad \e\to+0,
\end{equation}
as well as
\begin{equation*}%\l%abel{4.13}
U(x;\e)=\psi_0(x_2)+\Odr(\e)
\end{equation*}
in $\H^2(\Om_a)$ for each fixed $a>0$.
\end{lemma}
%

%---------------------------------------------------------%
\subsection{Proof of Theorems~\ref{th2.1} and \ref{th2.2}}\label{Sec.main}
%---------------------------------------------------------%
%
It follows from Lemma~\ref{lm4.3} that there is at most one simple
eigenvalue of~$H_\alpha$ converging to $\mu_0^2$ as $\e\to+0$. A
sufficient condition guaranteeing the existence of such eigenvalue
is the inequality
\begin{equation}\label{4.14}
\RE \big(k_1(\e)+\e k_2(\e)\big)\geqslant C(\e)\,\e^2, \qquad
C(\e)\to+\infty, \quad \e\to+0,
\end{equation}
that is implied by (\ref{4.5}), (\ref{4.12}), the definition of
the operator~$M_\e$ and the assumption on~$\b$.
The sufficient condition of the absence of the eigenvalue
is the opposite inequality
\begin{equation}\label{4.19a}
\RE \big(k_1(\e)+\e k_2(\e)\big)\leqslant -C(\e)\,\e^2, \qquad
C(\e)\to+\infty, \quad \e\to+0.
\end{equation}
Thus, we just need to calculate the numbers $k_1$ and $k_2$ to
prove the theorems.

It is easy to compute the coefficient $k_1$,
\begin{equation}\label{4.15}
  k_1(\e) =
  \begin{cases}
    \displaystyle
    -\alpha_0\la\b\ra + k_1'(0) \, \e
    & \mbox{if} \quad |\alpha_0|<\pi/d \,,
    \\
    k_1'(0) \, \e
    & \mbox{if} \quad |\alpha_0|>\pi/d \,,
  \end{cases}
\end{equation}
where
\begin{equation*}
  k_1'(0) := -\frac{1}{2} \int\limits_\Om
  \psi_0(x_2)\overline{\phi_0}(x_2)\big(\b^2(x_1)+\b'^2(x_1) \, x_2^2\big)\di x
  \,.
\end{equation*}

It is more complicated technically to calculate $k_2$. This
coefficient depends on $\e$ as well; to prove the theorem we need the
leading term of its asymptotics as $\e\to+0$. We begin the
calculations by observing an obvious identity,
\begin{equation*}
  L_\e U=L_0 U+\Odr(\e) \,,
\end{equation*}
where
\begin{align*}
(L_0 U)(x):=&-2\iu\b'(x_1)x_2\frac{\p U(x)}{\p x_1}-2\iu\b(x_1)
\frac{\p U(x) }{\p x_2}-\iu\b''(x_1) x_2U(x)
\\
=&\,\iu\Big(\b(x_1)x_2\D U(x)-\D\big[\b(x_1)x_2 U(x)\big]\Big),
\end{align*}
which is valid for each $U\in\Hloc^2(\Om)$ in $L_2(\Om_a)$, if
$a$ is large enough and independent of $\e$. Thus,
\begin{equation*}%\l%abel{4.20}
  k_2(\e)=k_2(0)+\Odr(\e),
  \qquad\mbox{where}\qquad
  k_2(0)=\frac{1}{2}\int\limits_{\Om}
  \overline{\phi_0}(x_2) \big(L_0 T_2(0) L_0\psi_0\big)(x) \di x.
\end{equation*}
We denote
$\widehat{U}:=T_2(0)\widehat{F}$ and $\widehat{F}:=L_0\psi_0$.
Taking into account the problem~(\ref{4.9}) for $\widehat{U}$
and integrating by parts, we obtain
\begin{align}\label{4.16}
%\begin{aligned}
  k_2(0)&=\frac{\iu}{2}\int\limits_\Om \overline{\phi_0}(x_2)
  \left(
  \b(x_1) x_2\D \widehat{U}(x)-\D \big[\b(x_1) x_2 \widehat{U}(x)\big]
  \right)\di x
  \nonumber \\
  &=-\frac{\iu}{2} \int\limits_\Om\overline{\phi_0}(x_2)
  (\D+\mu_0^2) \b(x_1) x_2 \widehat{U}(x)\di x
  -\frac{\iu}{2}\int\limits_\Om
  \overline{\phi_0}(x_2) \b(x_1) x_2 \widehat{F}(x) \di x
  \nonumber \\
  &=-\frac{\iu}{2} \left\langle
  \b\left[\overline{\phi_0}(d)\widehat{U}(\cdot,d)-
  \overline{\phi_0}(0)\widehat{U}(\cdot,0)\right]
  \right\rangle
  -\frac{\iu}{2}\int\limits_\Om
  \overline{\phi_0}(x_2) \b(x_1) x_2 \widehat{F}(x) \di x
  \,.
%\end{aligned}\label{4.16}
\end{align}
The last term on the right hand side of this identity is
calculated by integration by parts,
\begin{eqnarray*}
  \lefteqn{
  -\frac{\iu}{2}\int\limits_\Om
  \overline{\phi_0}(x_2) \b(x_1) x_2 \widehat{F}(x) \di x }
  \\
  &&= -\frac{1}{2} \int\limits_\Om \overline{\phi_0}(x_2)
  \b(x_1) x_2 \big[
  2\b(x_1) \psi_0'(x_2)+\b''(x_1) x_2\psi_0(x_2)
  \big] \di x
  \\
  &&=-\frac{1}{2}\int\limits_\Om
  \b^2(x_1) x_2
  \left(\overline{\phi_0}\psi_0\right)'(x_2) \di x
  +\frac{1}{2}\int\limits_\Om x_2^2 \b'^2(x_1)
  \overline{\phi_0}(x_2) \psi_0(x_2) \di x
  \\
  &&= -\frac{\overline{\phi_0}(d)\psi_0(d)d}{2} \la\b^2\ra
  +\frac{1}{2}\int\limits_\Om
  \overline{\phi_0}(x_2) \psi_0(x_2)
  \big[\b^2(x_1)+\b'^2(x_1) x_2^2\big] \di x
  \,.
\end{eqnarray*}
This formula, (\ref{4.15}) and (\ref{4.16}) yield
\begin{equation}\label{4.20b}
  k_1(\e)+\e k_2(\e)=
  \begin{cases}
  -\e\alpha_0\la\b\ra+\e^2 K+\Odr(\e^3)
  & \mbox{if} \quad |\alpha_0|<\pi/d \,,
  \\
  \e^2 K+\Odr(\e^3)
  & \mbox{if} \quad |\alpha_0|>\pi/d \,,
\end{cases}
\end{equation}
where
\begin{equation*}
  K := -\frac{\iu}{2}
  \left\langle
  \b \left[\overline{\phi_0}(d)\widehat{U}(\cdot,d)-
  \overline{\phi_0}(0)\widehat{U}(\cdot,0)\right]
  \right\rangle
  -\frac{\overline{\phi_0}(d)\psi_0(d)d}{2} \la\b^2\ra.
\end{equation*}
Thus, it remains to calculate~$K$.
In order to do it, we construct the function~$\widehat{U}$
as the series~(\ref{4.19}).

%---------------------------------------------%
\paragraph{Case \underline{$|\alpha_0|<\pi/d$}\,:}
%---------------------------------------------%
%
Using the identity
\begin{equation*}
  \widehat{F}(x) = -2\iu\b(x_1)\psi_0'(x_2)-\iu\b''(x_1) x_2\psi_0(x_2)
  \,,
\end{equation*}
one can check that
\begin{align*}
  \widehat{F}(x) &= -2\alpha_0\b(x_1)\psi_0(x_2)
  -\iu\b''(x_1) \sum\limits_{j=0}^{\infty}
  c_j \psi_j(x_2)
  \,,
  \\
  \widehat{U}_j(x_1) &=
  \begin{cases}
    \iu c_0 \b(x_1) - 2\alpha_0 v_0(x_1)
    & \mbox{if} \quad j=0 \,,
    \\
    \iu c_j
    \big[\b(x_1)-(\mu_j^2-\mu_0^2)v_j(x_1)\big]
    & \mbox{if} \quad j \geqslant 1 \,,
  \end{cases}
\end{align*}
where $c_j := \int_I x_2 \psi_0(x_2) \overline{\phi_j}(x_2) dx_2$
and the functions~$v_j$ were introduced in~(\ref{vjs}).
Substituting now the formulae for~$\widehat{U}_j$
and~(\ref{4.19}) into~(\ref{4.20b}),
we arrive at the following chain of identities
\begin{align*}
  K=&\frac{1}{2}
  \Big[
  \sum\limits_{j=0}^{\infty} c_j\psi_j(x_2) \overline{\phi}_0(x_2)
  \Big]_{x_2=0}^{x_2=d}
  \la\b^2\ra + 2\alpha_0^2 \la\b v_0\ra
  \\
  &-\frac{2\alpha_0}{d}\sum\limits_{j=1}^{\infty}
  \frac{\iu\big[\E^{\iu\alpha_0 d} -(-1)^j\big] \mu_j^2}
  {\big[\E^{\iu\alpha_0 d}+(-1)^j\big] (\mu_j^2-\mu_0^2)} \,
  \la\b v_j\ra
  - \frac{\overline{\phi}_0(d)\psi_0(d) d}{2}\la\b^2\ra
  \\
  =& \frac{1}{2}
  \Big[
  x_2\psi_0(x_2)\overline{\phi}_0(x_2)
  \Big]_{x_2=0}^{x_2=d}
  \la\b^2\ra
  +2\alpha_0^2\la\b v_0\ra
  \\
  & + \frac{2\alpha_0}{d}
  \sum\limits_{j=1}^{\infty} \frac{\mu_j^2\la \b
  v_j\ra}{\mu_j^2-\mu_0^2}\tan\frac{\alpha_0 d+\pi j}{2}
  -\frac{\overline{\phi}_0(d)\psi_0(d)d}{2} \la\b^2\ra
  \,,
\end{align*}
where the last expression coincides with~$\tau$ for
$|\alpha_0|<\pi/d$.

%---------------------------------------------%
\paragraph{Case \underline{$|\alpha_0|>\pi/d$}\,:}
%---------------------------------------------%
%
Following the same scheme as above, we arrive at
\begin{align*}
  \widehat{F}(x) =& -\iu\b''(x_1)
  \sum\limits_{j=0}^{\infty} c_j \psi_j(x_2)
  - \frac{4\alpha_0\psi_1(x_2)\b(x_1)}{1-\E^{-\iu\alpha_0d}}
  \\
  & -\frac{2\iu}{d}(\mu_0^2-\mu_1^2)
  \sum\limits_{j=1}^{\infty}\frac{\mu_{2j}^2\psi_{2j}(x_2)}
  {(\mu_{2j}^2-\mu_1^2) (\mu_{2j}^2-\mu_0^2)} \b(x_1)
  \,,
  \\
  \widehat{U}_j(x_1) =&
  \begin{cases}
    \iu c_0 \b(x_1)
    & \mbox{if} \quad j=0 \,,
    \\
    \displaystyle
    \iu c_1 \b(x_1)
    - \frac{2\alpha_0}{1-\E^{-\iu\alpha_0d}} v_1(x_1)
    & \mbox{if} \quad j=1 \,,
    \\
    \displaystyle
    \iu c_j \b(x_1)
    + \frac{2\iu\mu_j^2 \big[1+(-1)^j\big]}{(\mu_j^2-\mu_1^2)d} v_j(x_j)
    & \mbox{if} \quad j \geqslant 2 \,,
  \end{cases}
\end{align*}
and check that $K=\tau$ for $|\alpha_0|>\pi/d$.

The series in the formulae for~$\tau$ converge
since the functions~$v_j$ satisfy
\begin{equation}\label{4.20a}
  -v_j''+(\mu_j^2-\mu_0^2)v_j=\b
  \qquad\mbox{in}\quad\mathbb{R}
  \,,
\end{equation}
and by \cite[Ch. V, \S 3.5, Formula~(3.16)]{Kato}
($j \geqslant 1$)
\begin{equation}\label{1.5}
|\la \b v_j\ra|\leqslant
\|\b\|_{L_2(\mathbb{R})}\|v\|_{L_2(\mathbb{R})}\leqslant
\frac{\|\b\|_{L_2(\mathbb{R})}^2}{\mu_j^2-\mu_0^2}.
\end{equation}

Summing up,
\begin{equation*}
k_1(\e)+\e k_2(\e)=\left\{
\begin{aligned}
&-\e\alpha_0\la\b\ra+\e^2 \tau+\Odr(\e^3) && \text{if}\quad
|\alpha_0|<\pi/d \,,
\\
&\e^2 \tau+\Odr(\e^3) && \text{if}\quad |\alpha_0|>\pi/d \,.
\end{aligned}\right.
\end{equation*}
All the statements of the theorems
-- except for the reality of the eigenvalue --
follow from these formulae,
the identity $\l=\mu_0^2-k^2$,
the inequalities~(\ref{4.14}) and~(\ref{4.19a}),
Lemma~\ref{lm4.3}, the asymptotics~(\ref{4.5}) for~$U$,
and the definition of the operator~$M_\e$.

Let us show that~$\lambda_\e$ is necessarily real as $\e \to 0+$.
Let $\Psi_\e$ be the eigenfunction associated with the eigenvalue
$\l_\e$. It is easy to check that the function $x \mapsto
\overline{\Psi}_\e(x_1,d-x_2)$ is an eigenfunction of~$H_\alpha$
associated with the eigenvalue $\overline{\l}_\e$. This eigenvalue
converges to $\mu_0^2$ as $\e\to+0$. By the uniqueness of such
eigenvalue we obtain $\l_\e=\overline{\l}_\e$ that completes the
proof. \hfill\qed

%----------------------------------------%
\subsection{Proof of Theorem~\ref{th2.0}}
%----------------------------------------%
%
We employ here the same argument as in the previous proof. The
formula for $k(\e)$ in the case $\alpha_0=0$ can be obtained from
that for $|\alpha_0|<\pi/d$ by passing to the limit $\alpha_0\to0$.
It leads us to the relation
\begin{equation*}
  k(\e)=\e^2 \tau + \Odr(\e^3)
  \qquad\mbox{with}\qquad
  \tau=-\sum\limits_{j=0}^{\infty}
  \frac{4\la\b v_{2j+1}\ra}{\mu_{2j+1}d^2}
  \,.
\end{equation*}
To prove the theorem it is sufficient to show that $\tau<0$.
Indeed, the equation~(\ref{4.20a}) implies that for $j \geqslant 1$
\begin{equation}\label{4.24}
%\begin{aligned}
\la\b v_j\ra= \|v_j'\|_{L_2(\mathbb{R})}^2+
(\mu_j^2-\mu_0^2)\|v_j\|_{L_2(\mathbb{R})}^2>0
  \,.
%\end{aligned}
\end{equation}
Thus, $\tau<0$.
\hfill\qed

%--------------------------------------------%
\subsection{Proof of Proposition~\ref{th2.3}}
%--------------------------------------------%
%
Since $\la\b\ra=0$, the function~$v_0$ is constant at infinity.
Hence, by the equation~(\ref{4.20a}) for $v_0$,
\begin{equation*}
\la\b v_0\ra=\|v_0'\|_{L_2(\mathbb{R})}^2.
\end{equation*}
At the same time, it follows from~(\ref{1.5}) and~(\ref{4.24})
that for $j \geqslant 1$,
\begin{equation*}
0<\la \b v_j\ra
<\frac{\|\b\|_{L_2(\mathbb{R})}^2}{\mu_j^2-\mu_0^2}.
\end{equation*}
The relations obtained allow us to estimate
\begin{align*}
\tau \ > \ & 2\alpha_0^2\|v_0'\|_{L_2(\mathbb{R})}^2-
\frac{2\alpha_0\cot\frac{\alpha_0 d}{2}}{d}
\|\b\|_{L_2(\mathbb{R})}^2 \left(\frac{\mu_1^2}{(\mu_1^2-\mu_0^2)^2
}+\sum\limits_{j=1}^{\infty}
\frac{\mu_{2j+1}^2}{(\mu_{2j+1}^2-\mu_1^2)^2}\right)
\\
\ = \ & \frac{\alpha_0^2
l^3}{2}\bigg\|\int\limits_\mathbb{R}\sgn(\cdot-t_1)\widetilde{\b}
(t_1)\di t_1\bigg\|_{L_2(\mathbb{R})}^2
\\
\ & - \frac{2\alpha_0 l\cot\frac{\alpha_0 d}{2}}{d}
\left(\frac{\mu_1^2}{(\mu_1^2-\mu_0^2)^2}+\frac{d^2}{16\pi^2}
+\frac{d^2}{48}\right) \|\widetilde{\b}\|_{L_2(\mathbb{R})}^2
  \,,
\end{align*}
where the resulting expression is positive under the hypothesis.
\hfill\qed

%--------------------------------------------%
\subsection{Proof of Proposition~\ref{th2.4}}\label{Sec.final}
%--------------------------------------------%
%
Using~(\ref{4.24}), we obtain
\begin{equation*}
\tau>\frac{2\alpha_0\pi^2\cot\frac{\alpha_0
d}{2}}{(\mu_1^2-\mu_0^2)d^3} \la\b
v_1\ra+\frac{8\pi^2}{(\mu_1^2-\mu_0^2)d^4}
\sum\limits_{j=1}^{m}\frac{\mu_{2j}^2 \la\b
v_{2j}\ra}{\mu_{2j}^2-\mu_1^2}
  \,,
\end{equation*}
where the right hand side is non-negative under the hypothesis~(\ref{1.8a}).
\hfill\qed

\medskip
Let us show that the inequality~(\ref{1.8a}) can be achieved if
$\alpha_0\to \mu_2+0$. In this case $m=1$, and it is sufficient to
check that
\begin{equation}\label{4.25}
\alpha_0\la \b v_1\ra\cot\frac{\alpha_0 d}{2}>\frac{16\pi^2\la\b v_2
\ra}{(\mu_1^2-\mu_2^2)d^3}=\frac{16\pi^2\la\b v_2
\ra}{(\alpha_0^2-\mu_2^2)d^3}.
\end{equation}
It follows from the definition of~$v_1$ that it satisfies the
asymptotic formula
\begin{equation*}
  v_1(x_1) = v_2(x_1) + (\mu_2-\alpha_0)\,\widehat{v}(x_1)
  +\Odr\big((\mu_2-\alpha_0)^2\big)
\end{equation*}
in $L_2(\mathbb{R})$-norm,
where the function~$\widehat{v}$ is given by
\begin{equation*}
\widehat{v}(x_1):=\int\limits_\mathbb{R} \frac{(\sqrt{3} \mu_0
|x_1-t_1|+1)\E^{-\sqrt{3}\mu_0|x_1-t_1|}}{3\sqrt{3}\mu_0^2}
\b(t_1)\di t_1,
\end{equation*}
and satisfies the equation
\begin{equation*}
  -\widehat{v}''+3\mu_0^2\,\widehat{v}=4 \mu_0 v_2
  \qquad\mbox{in}\quad\mathbb{R}.
\end{equation*}
We multiply this equation by $v_2$ and integrate by parts over
$\mathbb{R}$ taking into account the equation~\eqref{4.20a}
for~$v_2$,
\begin{equation*}
4\mu_0\|v_2\|_{L_2(\mathbb{R})}^2=\la \b \widehat{v}\ra.
\end{equation*}
Hence,
\begin{equation*}
\la\b v_1\ra=\la \b v_2\ra+4(\mu_2-\alpha_0)\mu_0
\|v_2\|_{L_2(\mathbb{R})}^2+ \Odr\big((\mu_2-\alpha_0)^2\big).
\end{equation*}
Employing this identity, we write the asymptotic expansions for the
both sides of~\eqref{4.25} as $\alpha_0\to\mu_2+0$, and obtain
\begin{align*}
  \alpha_0\la \b v_1\ra\cot\frac{\alpha_0 d}{2}
  &=\frac{4\pi\la\b v_2\ra}{(\alpha_0-\mu_2)d^2}
  +\frac{2}{d}\left(\la\b v_2\ra-\frac{8\pi^2}{d^2}\|v_2\|_{L_2(\mathbb{R})}^2\right)
  +\Odr(\mu_2-\alpha_0) \,,
  \\
  \frac{16\pi^2\la\b v_2 \ra}{(\alpha_0^2-\mu_2^2)d^3}
  &=\frac{4\pi\la\b v_2\ra}{(\alpha_0-\mu_2)d^2}
  -\frac{\la\b v_2\ra}{d}+\Odr(\mu_2-\alpha_0) \,.
\end{align*}
Thus, to satisfy~\eqref{4.25}, it is sufficient to check that
\begin{equation*}
3\la\b v_2\ra-\frac{16\pi^2}{d^2}\|v_2\|_{L_2(\mathbb{R})}^2>0
  \,,
\end{equation*}
which is in view of \eqref{4.24} equivalent to
\begin{equation*}
3\|v'_2\|_{L_2(\mathbb{R})}^2
>\frac{7\pi^2}{d^2}\|v_2\|_{L_2(\mathbb{R})}^2
  \,.
\end{equation*}
It is clear that there exists a function $v\in
C_0^\infty(\mathbb{R})$ for which this inequality is valid. Letting
$v_2:=v$ and $\b:=-v''+3 \mu_0^2 v$, we conclude that there exists
$\b$ such that the inequality \eqref{1.8a} holds true, if $\alpha_0$
is sufficiently close to $\mu_2$ and greater than this number.

\section*{Acknowledgment}
The authors are grateful to Miloslav Znojil for many valuable
discussions. They also thank the referee for helpful
remarks and suggestions.

The work was partially supported by the Czech Academy of Sciences
and its Grant Agency within the projects IRP AV0Z10480505 and
A100480501, and by the project LC06002 of the Ministry of Education,
Youth and Sports of the Czech Republic. D.B. was also supported by
RFBR (07-01-00037), by Marie Curie International Fellowship within
6th European Community Framework Programme (MIF1-CT-2005-006254); in
addition he gratefully acknowledges the support from Deligne 2004
Balzan Prize in mathematics and the grant of the President of
Russian Federation for young scientists and their supervisors
(MK-964.2008.1) and for Leading Scientific Schools
(NSh-2215.2008.1). D.K. was also supported by FCT, Portugal, through
the grant SFRH/\-BPD/\-11457/\-2002.

%\newpage
%--------------%
% BIBLIOGRAPHY %
%--------------%
%
{\small
%\bibliography{bib}
%\bibliographystyle{amsplain}

\providecommand{\bysame}{\leavevmode\hbox to3em{\hrulefill}\thinspace}
\providecommand{\MR}{\relax\ifhmode\unskip\space\fi MR }
\providecommand{\MRhref}[2]{%
  \href{http://www.ams.org/mathscinet-getitem?mr=#1}{#2}
}
\providecommand{\href}[2]{#2}

}

\end{document}